\documentclass[12pt,a4paper]{article}

\usepackage{amsmath, amsthm, amssymb,accents,mathrsfs,mathtools}
\usepackage{enumerate}
\usepackage{pdflscape}
\usepackage{authblk}

\usepackage[pdftex,plainpages=false,hypertexnames=false,pdfpagelabels]{hyperref}
\usepackage{xcolor}
\definecolor{dark-red}{rgb}{0.7,0.25,0.25}
\definecolor{dark-blue}{rgb}{0.15,0.15,0.55}
\definecolor{medium-blue}{rgb}{0,0,.8}
\definecolor{DarkGreen}{RGB}{0,150,0}
\definecolor{rho}{named}{red}
\hypersetup{
   colorlinks, linkcolor={purple},
   citecolor={medium-blue}, urlcolor={medium-blue}
}

\usepackage{fullpage}
\theoremstyle{plain}
\newtheorem{thm}{Theorem}[section]
\newtheorem{cor}[thm]{Corollary}

\newtheorem{lem}[thm]{Lemma}
\newtheorem{prop}[thm]{Proposition}
\theoremstyle{definition}
\newtheorem{defn}[thm]{Definition}
\theoremstyle{remark}

\newtheorem{rem}[thm]{Remark}

\numberwithin{equation}{section}

\DeclareMathOperator{\dom}{Dom}
\DeclareMathOperator{\End}{End}
\DeclareMathOperator{\Hom}{Hom}

\DeclareMathOperator{\Lie}{Lie}

\DeclareMathOperator{\Ran}{Ran}
\DeclareMathOperator{\Mob}{M\ddot{o}b}
\DeclareMathOperator{\spann}{span}
\DeclareMathOperator{\supp}{supp}

\newcommand{\abs}[1]{\left| #1 \right|}
\newcommand{\ip}[1]{\langle #1 \rangle}
\newcommand{\norm}[1]{\left\| #1 \right\|}
\newcommand{\set}[2]{\left\{#1 \, \middle|\, #2\right\}}
\newcommand{\comm}[2]{[ #1, #2 ]}
\newcommand{\pd}{\partial}

\def\semicolon{;}
\def\applytolist#1{
    \expandafter\def\csname multi#1\endcsname##1{
        \def\multiack{##1}\ifx\multiack\semicolon
            \def\next{\relax}
        \else
            \csname #1\endcsname{##1}
            \def\next{\csname multi#1\endcsname}
        \fi
        \next}
    \csname multi#1\endcsname}

\def\calc#1{\expandafter\def\csname c#1\endcsname{{\mathcal #1}}}
\applytolist{calc}QWERTYUIOPLKJHGFDSAZXCVBNM;
\def\bbc#1{\expandafter\def\csname bb#1\endcsname{{\mathbb #1}}}
\applytolist{bbc}QWERTYUIOPLKJHGFDSAZXCVBNM;
\def\bfc#1{\expandafter\def\csname bf#1\endcsname{{\mathbf #1}}}
\applytolist{bfc}QWERTYUIOPLKJHGFDSAZXCVBNM;
\def\sfc#1{\expandafter\def\csname s#1\endcsname{{\sf #1}}}
\applytolist{sfc}QWERTYUIOPLKJHGFDSAZXCVBNM;

\newcounter{wightman}

\title{Unitary vertex algebras and Wightman conformal field theories}
\author[1]{Christopher Raymond\thanks{{\tt Christopher.Raymond@anu.edu.au}}}
\author[2]{Yoh Tanimoto\thanks{{\tt hoyt@mat.uniroma2.it}}}
\author[1]{James E.\! Tener\thanks{{\tt James.Tener@anu.edu.au}}}

\affil[1]{Mathematical Sciences Institute, The Australian National University,\authorcr
Canberra, ACT 2600, Australia}
\affil[2]{Dipartimento di Matematica, Universit\`a di Roma Tor Vergata,\authorcr
   Via della Ricerca Scientifica 1, I-00133 Roma, Italy}
\date{}

\begin{document}

\maketitle

\begin{abstract}
We prove an equivalence between the following notions: (i) unitary M\"obius vertex algebras, and (ii) Wightman conformal field theories on the circle (with finite-dimensional conformal weight spaces) satisfying an additional condition that we call uniformly bounded order.
Reading this equivalence in one direction, we obtain new analytic and operator-theoretic information about vertex operators.
In the other direction we characterize OPEs of Wightman fields and show they satisfy the axioms of a vertex algebra.
As an application we establish new results linking unitary vertex operator algebras with conformal nets.
\end{abstract}

{\hypersetup{linkcolor=black}
\setcounter{tocdepth}{2}
\tableofcontents
}

\section{Introduction}

Two-dimensional conformal field theory has attracted the attention of physicists and mathematicians for its
rich algebraic and analytic structures (see e.g.\! \cite{DMS97}).
There are several well-developed approaches to axiomatizing unitary two-dimensional chiral conformal field theories (chiral CFTs), and in recent years the problem of reconciling the various approaches has garnered considerable interest.
In this article, we will compare two formulations of two-dimensional unitary chiral conformal field theories, namely the vertex algebra approach
(in terms of formal power series)
and the Wightman approach (in terms of operator-valued distributions on the circle).
While these two axiomatizations are supposed to describe the same theories, the manners in which they do so are quite different.
In particular, Wightman fields rely on nontrivial analytic properties which are not readily visible in the vertex algebra formulation.

Our main result is that unitary M\"obius vertex algebras are equivalent to M\"obius-covariant Wightman field theories that satisfy an analytic condition which we call \emph{uniformly bounded order}.
The adjective ``M\"obius'' describes the underlying symmetry of the theories under consideration, which is $\mathfrak{sl}_2$ at the infinitesimal level or $\mathrm{PSU}(1,1)$ at the global level.
The main result is summarized in Theorem \ref{thm: wightman CFT iff va}.
A highlight of the result is that we do not impose any extraneous constraints on the vertex algebras under consideration, which is a departure from other results comparing axiomatizations of chiral CFTs.
We establish the main result by giving a careful axiom-by-axiom comparison of the two notions, with emphasis on the key properties of locality and M\"obius covariance.

The uniformly bounded ordered property arises in our treatment of the locality axiom, and it seems that this property had not been noticed in earlier comparisons of vertex operators and Wightman fields at a physical level of rigor (e.g. \cite[\S1.2]{Kac98}; see Section \ref{subsec: VA to W and back} for further discussion).
We use the locality axiom of vertex algebras to show that unitary M\"obius vertex algebras satisfy uniformly bounded order, which yields new analytic information about these theories.
On the other hand, we observe that uniformly bounded order for Wightman fields allows us to construct a vertex algebra structure.
 This amounts to showing operator product expansion for Wightman fields (cf.\! \cite{Bostelmann05OPE}).

As an application, we are able to obtain further results comparing unitary vertex operator algebras (VOA) with conformal nets, another well-studied axiomatization of two-dimensional unitary chiral conformal field theories in the spirit of the Haag-Kastler axioms for algebraic quantum field theory (AQFT).
First, starting from a conformal net we construct a canonical unitary VOA structure on a subspace of the finite-energy vectors of the theory.
We show that this subspace is, in fact, the entire space of finite-energy vectors whenever the conformal net ``comes from'' a unitary VOA in a loose sense.
This construction applies to arbitrary conformal nets, extending a result of \cite{CKLW18} (based on the construction of \cite{FJ96}).
Second, we construct conformal nets from unitary VOAs possessing a strong commutativity property, which reproduces a result of \cite{CKLW18} without requiring the technical assumption of ``polynomial energy bounds.''

The article is organized as follows:
In Section \ref{sec: two formulations} we describe the vertex algebra and Wightman formulations of unitary two-dimensional chiral conformal field theories.
We also give preliminary results constructing Wightman fields from unitary vertex algebras.
In Section \ref{sec:vertexwightman} we give a careful comparison of the two different axiomatizations, with particular attention to locality and M\"obius covariance. Here, we prove our main result giving an equivalence between the two notions.
In Section \ref{sec: conformal nets} we discuss conformal nets and give several applications of our results in this context.

\subsection*{Acknowledgements}

CR and JT were supported by ARC Discovery Project DP200100067.
CR was also supported by ARC Discovery Project DP170103265.
YT was supported by Programma per giovani ricercatori, ``Rita Levi Montalcini''
of the Italian Ministry of Education, University and Research until 2020.

YT acknowledges the MIUR Excellence Department Project awarded to
the Department of Mathematics, University of Rome ``Tor Vergata'' CUP E83C18000100006 and the University of
Rome ``Tor Vergata'' funding scheme ``Beyond Borders'' CUP E84I19002200005.
JT gratefully acknowledges the support of ``Rita Levi Montalcini'' grant and the hospitality of the University of Rome, Tor Vergata, which facilitated this collaboration.

\section{Two formulations of chiral CFT}\label{sec: two formulations}
In this section we introduce our two formulations of chiral CFT, namely unitary M\"obius vertex algebras and M\"obius-covariant Wightman field theories.

\subsection{Positive-energy representations of the M\"obius group}\label{sec: PERs}

We denote by $\Mob$ the group of holomorphic automorphisms of the unit disk, which we identify with fractional-linear transformations via
\[
\Mob \cong \mathrm{PSU}(1,1) = \left\{ \begin{pmatrix} a & b\\ \overline{b} & \overline{a}\end{pmatrix} \, \, : \, \, a,b \in \bbC, \, |a|^2 - |b|^2 = 1 \right\} / \left\{ \pm \mathrm{Id} \right\}.
\]
It is a three-dimensional real Lie group with Lie algebra denoted $\Lie(\Mob)$.
If we consider $\Mob$ as a subgroup of the group $\mathrm{Diff}(S^1)$ of orientation-preserving diffeomorphisms of the unit circle $S^1 \subset \bbC$, $\Lie(\Mob)$ is identified with
a three-dimensional subspace of the space of smooth vector fields $\mathrm{Vect}(S^1)$ on $S^1$.
Each vector field is identified with a differential operator $f(e^{i\theta})\frac{d}{d\theta}$ for some smooth function $f(e^{i\theta})$,
and the Lie bracket is given by $[f\frac{d}{d\theta},g\frac{d}{d\theta}] = (f'g - fg')\frac{d}{d\theta}$, where $f'$ denotes $\tfrac{df}{d\theta}$.
Note that this bracket is the opposite of the bracket of vector fields, which is the natural choice when identifying $\mathrm{Vect}(S^1)$ with the Lie algebra of $\mathrm{Diff}(S^1)$.
This gives rise to an embedding $\Lie(\Mob) \subset \mathrm{Vect}(S^1)$, and the complexification $\Lie(\Mob)_\bbC \cong \mathfrak{sl}(2,\bbC)$ of $\Lie(\Mob)$ is spanned by the elements $\{L_{-1},L_0,L_1\}$,
where $L_m$ is the complexified vector field $-ie^{im\theta}\frac{d}{d\theta}$.
The vector fields $L_m$ satisfy the commutation relations
\begin{align*}
 [L_m, L_n] = (m-n)L_{m+n}, \qquad m,n=-1,0,1.
\end{align*}

A representation of $\Lie(\Mob)$ on an inner product space $\cV$ is called a \textbf{unitary representation} if $\ip{L_m u, v} = \ip{u, L_{-m}v}$ for all $u,v \in \cV$ and $m=-1,0,1$.
Here we have extended the representation of $\Lie(\Mob)$ to its complexification, and, as an abuse of notation, used the same symbols $L_m$ for operators on the representation space.
We are interested in \textbf{positive-energy representations} of $\Lie(\Mob)$, which are unitary representations with the additional property that $L_0$ is a diagonalizable operator with non-negative eigenvalues.
Any positive-energy representation of $\Lie(\Mob)$ in which $L_0$ has integral spectrum can be integrated to a continuous (in the strong operator topology, SOT) unitary representation of $\Mob$ on the Hilbert space completion $\cH$ of $\cV$ (see e.g.\! \cite{Longo08}, \cite[Appendix A]{WeinerThesis}, or \cite[\S3.1]{CKLW18}).

Conversely, an SOT-continuous unitary representation $\cH$ of $\Mob$ is called a \textbf{positive-energy representation} if the generator of the rotation subgroup is a positive operator (with necessarily integral spectrum). 
In this case, the representation is the integration of a positive-energy representation of $\Lie(\Mob)$ on the dense subspace $\cV$ of finite-energy vectors (that is, the algebraic span of eigenvectors of the generator of rotation).
In this article, we will emphasize positive-energy representations in which the generator of rotation $L_0 \in \Lie(\Mob)_\bbC$ has finite-dimensional eigenspaces.

\subsection{Unitary M\"obius vertex algebras}\label{sec:vertex}

In this section, we define \textbf{(unitary) M\"obius vertex algebras}, which were introduced under the name quasi-vertex operator algebra in \cite{FHL93}. 
The reader is cautioned that other authors have used the term M\"obius vertex algebra to refer to a more general notion, in particular relaxing the requirement of finite-dimensionality of the weight spaces (e.g.\! \cite{BK08}).

If $\cV$ is a vector space, we denote by $\End(\cV)[[z^{\pm 1}]]$ the set of formal power series in $z^{\pm 1}$ with coefficients in $\End(\cV)$.
Given $v \in \cV$ and $A(z) =\sum_{n \in \bbZ} A_n z^n \in \End(\cV)[[z^{\pm 1}]]$, we can consider a formal series $A(z)v = \sum_n A_n v z^n$ in the variables $z^{\pm 1}$
with coefficients in $\cV$. Furthermore, for any $B \in \End(\cV)$ we have
$[A(z), B] = \sum_n [A_n, B]z^n \in \End(\cV)[[z^{\pm 1}]]$. If $B(w)$ is another formal series,
then the expression $[A(z), B(w)]$ makes sense as a formal series in $z^{\pm 1}$ and $w^{\pm 1}$.

\begin{defn}
An ($\bbN$-graded) \textbf{M\"obius vertex algebra} consists of a vector space $\cV$ equipped with a representation $\{L_{-1},L_0,L_1\}$ of $\Lie(\Mob)_\bbC$, a state-field correspondence $Y:\cV \to \End(\cV)[[z^{\pm 1}]]$,
and a choice of vector $\Omega \in \cV$ such that the following hold:
\begin{enumerate}[{(VA}1{)}]
\item $\cV = \bigoplus_{n=0}^\infty \cV(n)$, where $\cV(n) =  \ker(L_0 - n)$ and each $\cV(n)$ is finite-dimensional.
\item $\Omega$ is $\Lie(\Mob)$-invariant, i.e.\! $L_k \Omega = 0$ for $k=-1,0,1$, and $Y(\Omega,z) = \operatorname{Id}_\cV$.
\item $Y(v,z)\Omega$ has only non-negative powers of $z$ for all $v \in \cV$.
\item $[L_k, Y(v,z)] = \sum_{j=0}^{k+1} \binom{k+1}{j}z^{k+1-j} Y(L_{j-1}v,z)$ and $Y(L_{-1}v,z) = \frac{d}{dz} Y(v,z)$ for all $v \in \cV$ and $k = -1,0,1$.
\item $(z-w)^N[Y(v,z), Y(u,w)] = 0$ for $N$ sufficiently large.
\end{enumerate}
\end{defn}

For $v \in \cV$, the field $Y(v,z)$ associated with $v$ can be expanded as a formal series $Y(v,z) = \sum_{m \in \bbZ} v_{(m)}z^{-m-1}$, where $v_{(m)} \in \End(\cV)$.
A vector $v \in \cV$ is called \textbf{homogeneous (with conformal dimension $d$)} if it lies in $\cV(d)$. As a direct consequence of the axioms, when $v$ is homogeneous with conformal dimension $d$ we have
\begin{align*}
[L_0, Y(v,z)] &= zY(L_{-1}v,z) + Y(L_{0}v, z) = z\tfrac{d}{dz}Y(v,z) + dY(v,z) \\
&= \sum_{m \in \bbZ} (d-m-1) v_{(m)}z^{-m-1}.
\end{align*}
By comparison of formal series, one can deduce that $[L_0, v_{(m)}] = (d-m-1)v_{(m)}$. It follows that $v_{(m)}$ maps $\cV(n)$ into $\cV(n+d-m-1)$.

An (antilinear) \textbf{automorphism} of $\cV$ is an (anti)linear bijection $\Theta:\cV  \to \cV$ that fixes $\Omega$, commutes with the $L_m$, and satisfies
$\Theta(v_{(n)}u) = \Theta(v)_{(n)}\Theta(u)$ for all $v,u \in \cV$ and $n \in \mathbb{Z}$.
A vector $v \in \cV$ is called \textbf{quasi-primary} if it is homogeneous and $L_{1}v=0$. The field $Y(v,z)$ associated with a quasi-primary vector $v$ is called a quasi-primary field.

An invariant bilinear form $( \cdot , \cdot)$ on a M\"obius vertex algebra $\cV$ is a bilinear form such that 
\[
(Y(v,z)u,u') = (u, Y(e^{zL_1}(-z^{-2})^{L_0}v,z^{-1})u').
\]
\begin{defn}
A \textbf{unitary M\"obius vertex algebra} is a M\"obius vertex algebra $\cV$ equipped with an inner product $\ip{\cdot, \cdot}$
and an antilinear automorphism $\Theta$ such that the following hold:
\begin{enumerate}
\item $\ip{L_{-n}v,u} = \ip{v,L_{n}u}$ for all $v,u \in \cV$ and $n=-1,0,1$ (i.e. $\cV$ is a positive-energy representation of $\Lie(\Mob)_\bbC$).
\item $\norm{\Omega} = 1$.
\item $\ip{ \cdot, \Theta \cdot}$ is an invariant bilinear form
\end{enumerate}
\end{defn}
Note that we follow the convention that the inner product is linear in the first variable.
As in \cite[\S5]{CKLW18}, the operator $\Theta$ is automatically an involution, and a unitary M\"obius vertex algebra is simple (has no ideals) if and only if $\cV(0) = \bbC \Omega$.
Moreover if $\cV(0) = \bbC \Omega$, then there is at most one involution $\Theta$ making $\cV$ into a unitary M\"obius vertex algebra, and so in this case we can omit $\Theta$ from the extra structure needed to specify unitarity.

If $v \in \cV$ is homogeneous with conformal dimension $d$, then we write $v_{n}:=v_{(n+d-1)}$ for the degree-shifted mode,
and we extend this definition to arbitrary $v \in \cV$ by linearity.
The shifted modes have the property that $v_n \cV(m) \subset \cV(m-n)$.
For convenience of calculation, we will make use of both conventions throughout the paper.
As in \cite{CKLW18}, we call a quasi-primary field $Y(v,z)$ Hermitian if $\ip{v_n u,u'} = \ip{u,v_{-n}u'}$ for every $u,u' \in \cV$ and $n \in \bbZ$.
This last condition is equivalent to requiring $\Theta v = (-1)^{d_v} v$, where $d_v$ is the conformal dimension of $v$.

Let $\mathcal{S}\subset \cV$ be a set of homogeneous vectors. 
We say that the fields $\{Y(v,z) \, : \, v \in \mathcal{S}\}$ generate $\cV$
if
\[
\cV = \mathrm{Span}\{(v_1)_{(n_1)}\cdots (v_k)_{(n_k)} u \, : \, v_j,u \in \mathcal{S}, \, n_j\in \bbZ, \, k \in \bbZ_{\geq 0} \}.
\]
We have the following version of \cite[Proposition 5.17]{CKLW18} for M\"obius vertex algebras, proven exactly the same way as for vertex operator algebras in the reference.

\begin{prop}\label{prop: unitary iff Hermitian generating}
Let $\cV$ be a $\bbN$-graded M\"obius vertex algebra with $\cV(0) = \bbC \Omega$.
Suppose that $\cV$ is equipped with an inner product such that $\norm{\Omega} = 1$ and such that the representation of $\Lie(\Mob)_\bbC$ on $\cV$ is unitary.
Then $\cV$ is a unitary M\"obius vertex algebra if and only if it possess a generating family of Hermitian quasi-primary fields.
\end{prop}

\subsection{M\"obius-covariant Wightman fields}\label{sec:wightman}
We now give an alternate axiomatization of two-dimensional chiral CFTs, in terms of Wightman fields. 
The Wightman axioms are usually defined on the Minkowski space $\mathbb{R}^{d+1}$, but they can be adapted to chiral conformal field theory.

Let $\cH$ be a Hilbert space, and let $\cD \subseteq \cH$ be a dense subspace.
An \textbf{operator-valued distribution} on the unit circle $S^1$ with domain $\mathcal{D}$ is a continuous linear map 
\[
\varphi:C^\infty(S^1) \to \Hom(\cD,\cH),
\]
where $\Hom(\cD,\cH)$ is the space of linear maps $\cD \to \cH$, equipped with the topology of pointwise convergence.
We regard $\varphi(f)$ as an unbounded operator%
\footnote{Following standard references, the domain $\dom(A)$ of an unbounded operator $A$ is treated as part of the data of $A$.
If $B$ is the restriction of $A$ to a subspace $\dom(B)  \subset \dom(A)$,
then we write $B \subset A$. The adjoint $A^*$ of $A$ is defined on $\Psi \in \cH$ for which there exists a constant $C_A$
 such that $\|\ip{A\Phi, \Psi}\| \le C_A \|\Phi\|$  for all $\Phi \in \dom(A)$,
and it is characterized by $\ip{\Phi, A^*\Psi} = \ip{A\Phi, \Psi}$.
An unbounded operator is called closed if its graph is a closed subspace of $\cH \times \cH$, and called closable if the closure of its graph is the graph of some (closed) operator.
See e.g.\! \cite{Rudin91,Pedersen89} for basic notions of unbounded operators.
}
 on $\cH$ with domain $\cD$.

An operator-valued distribution $\varphi$ with domain $\cD$ is called \textbf{adjointable} (or an \textbf{adjointable field}) on $\cD$ if
there exists an operator-valued distribution $\varphi^\dagger$ with domain $\cD$ such that
$$
\ip{\varphi(f)\Phi,\Psi} = \ip{\Phi,\varphi^\dagger(\overline{f})\Psi}
$$
for all $\Phi,\Psi \in \cD$.
Equivalently, $\varphi$ is adjointable if $\cD \subseteq \dom(\varphi(f)^*)$ for every $f \in C^\infty(S^1)$, in which case the adjoint distribution is given by $\varphi^\dagger(f) := \varphi(\overline{f})^*\mid_\cD$.
Note that continuity of $f \mapsto \varphi^\dagger(f)\Phi$ is automatic, by the closed graph theorem.
If $\varphi$ is adjointable, then so is $\varphi^{\dagger}$ and $\varphi^{\dagger\dagger} = \varphi$.
It follows that when $\varphi$ is adjointable the operators $\varphi(f)$ and $\varphi^\dagger(f)$ are closable.
Note that an adjointable field with domain $\cD$ remains adjointable when restricted to a dense subspace $\cD^\prime \subset \cD$.

We now assume that our Hilbert space carries a positive-energy representation of $\Mob$.
Following \cite[\S 6]{CKLW18}, for $\gamma \in \Mob$ we denote by $X_\gamma \in C^\infty(S^1)$ the function
\[
X_\gamma(e^{i\theta}) = -i \frac{d}{d \theta} \log(\gamma(e^{i \theta})),
\]
which takes values in the positive real numbers since $\gamma$ is an orientation preserving diffeomorphism of $S^1$.
For $f \in C^\infty(S^1)$ and $d \in \bbZ_{\ge 0}$ we denote by $\beta_d(\gamma)f \in C^\infty(S^1)$ the function
\begin{equation}\label{eq:beta}
(\beta_d(\gamma)f)(z) = (X_\gamma(\gamma^{-1}(z)))^{d-1} f(\gamma^{-1}(z)).
\end{equation}

Now, let $\cD$ be a dense subspace of $\cH$. An operator-valued distribution with domain $\cD$ is called \textbf{M\"obius-covariant of degree $d$} (in the Wightman sense)
if for every $\gamma \in \Mob$ we have $U(\gamma)\cD = \cD$ and 
\[
U(\gamma)\varphi(f)U(\gamma)^* = \varphi(\beta_d(\gamma)f)
\]
for every $f \in C^\infty(S^1)$, where the equality is that of maps $\cD \to \cH$.

\begin{rem}\label{rem: tensor field}
The function $X_\gamma \in C^\infty(S^1)$ is characterized in terms of the pushforward of the vector field $\tfrac{d}{d\theta} \in \Gamma(TS^1)$ by the identity $\gamma_* \tfrac{d}{d\theta} = (X_\gamma \circ \gamma^{-1}) \tfrac{d}{d\theta}$.
It follows that when $d > 0$ the operation $\beta_d(\gamma)$ can be understood in terms of pushforwards of sections of the tensor product bundle $T^{\otimes d-1}S^1$:
\[
\gamma_* \left(f \tfrac{d}{d\theta}^{\otimes d-1}\right) = \left(\beta_d(\gamma)f\right) \tfrac{d}{d\theta}^{\otimes d-1}.
\]
Thus in some contexts the input to a Wightman field of degree $d$ is best understood as a tensor field $f \tfrac{d}{d\theta}^{\otimes d-1}$ of degree $d-1$ rather than a function $f$ (see \cite[Chapter 12]{Lee13SmoothManifolds}).
Indeed, the same holds when $d=0$, with $\tfrac{d}{d\theta}^{\otimes -1}$ interpreted as $d\theta$.
\end{rem}

A set $\mathcal{F}$ of operator-valued distributions on $S^1$ with a common dense domain $\mathcal{D} \subset \mathcal{H}$, together with a SOT-continuous unitary representation $U$ of $\Mob$ and a vector $\Omega \in \cD$, are called a \textbf{M\"obius-covariant Wightman field theory} on $S^1$ if they satisfy the following axioms:
\begin{enumerate}[{(W}1{)}]
\item \textbf{M\"obius covariance}: For every $\gamma \in \Mob$, $U(\gamma)$ preserves $\mathcal{D}$. 
For each $\varphi \in \mathcal{F}$ there is $d \in \mathbb{Z}_{\ge 0}$ such that $\varphi$ is M\"obius-covariant of degree $d$.\label{itm: W Mob covariance}
\item \textbf{Adjoint}: Each $\varphi \in \mathcal{F}$ is adjointable and $\cF$ is closed under the involution sending $\varphi$ to $\varphi^\dagger$.\label{itm: W adjoint}
\item \textbf{Locality}: If $f$ and $g$ have disjoint supports, then $\varphi_1(f)$ and $\varphi_2(g)$ commute
for any pair $\varphi_1, \varphi_2 \in \mathcal{F}$.
\label{itm: W locality}
\item \textbf{Spectrum condition}: The generator $L_0$ of rotations $U(R_\theta)=e^{i\theta L_0}$ is positive.
\label{itm: W spec}
\item \textbf{Vacuum}: The vector $\Omega$ is the unique (up to scalar) vector that is invariant under $U$, and $\Omega$ is cyclic for $\cF$ in the following sense:
whenever $\varphi_1, \ldots, \varphi_k \in \cF$ and $f_1,\ldots,f_k \in C^\infty(S^1)$ we have $\varphi_k(f_k) \cdots \varphi_1(f_1)\Omega \in \mathcal{D}$, and moreover $\mathcal{D}$ is spanned by vectors of this form.
\label{itm: W Vacuum}
\setcounter{wightman}{\value{enumi}}
\end{enumerate}
The vacuum axiom implies that $\cD$ is invariant under the operators $\varphi(f)$.

It follows automatically that expressions such as $\varphi_1(f_1) \cdots \varphi_k(f_k)\Omega$ are jointly continuous in the $f_j$.
Indeed, $f_j \mapsto \varphi(f_1) \ldots \varphi(f_k)\Omega$ has a closed graph, because
$f_j \mapsto \ip{\varphi(f_j)\Phi, \Psi}$ is continuous for any $\Phi, \Psi \in \mathcal{D}$ and $\mathcal{D}$ is dense. 
As $C^\infty(S^1)$ is a Fr\'echet space, the map is therefore continuous.
Hence the expression $\varphi(f_1) \ldots \varphi(f_k)\Omega$ is separately continuous in the $f_j$, and separately continuous maps from Fr\'echet spaces into a Fr\'echet space (the Hilbert space $\cH$) are jointly continuous
(see e.g.\! \cite[Theorem 2.17]{Rudin91}).

Let $\cV$ be the finite-energy vectors $\cV = \bigoplus_{n=0}^\infty \cV(n)$, where $\cV(n) = \operatorname{ker}(L_0-n)$ and the direct sum is the algebraic direct sum.
In this article we will emphasize theories that also satisfy:

\begin{enumerate}[{(W}1{)}]
\setcounter{enumi}{\value{wightman}}
\item \label{itm: W V(n) finite dim} The eigenspaces $\cV(n)=\mathrm{ker}(L_0-n)$ of $L_0$ are finite-dimensional.
\setcounter{wightman}{\value{enumi}}
\end{enumerate}
Examples of Wightman field theories that do not satisfy (W\ref{itm: W V(n) finite dim}) can be produced by the infinite tensor product construction, as in \cite[\S6]{CW05}.

We claim that if a Wightman field theory $\cF$ satisfies (W\ref{itm: W V(n) finite dim})  then $\cV \subset \cD$.
First, from the M\"obius covariance axiom we see that the subspace $\cW$ of $\cD$ given by
\[
\cW = \operatorname{span}\{ \, \varphi_k(f_k) \cdots \varphi_1(f_1)\Omega \,\, : \,\, f_j \in \bbC[z^{\pm 1}], \, \varphi_j \in \cF, \, k \in \bbZ_{\ge 0}\}
\]
is contained in $\cV$ and invariant under $L_0$.
Thus $\cW = \bigoplus_{n=0}^\infty \cW(n)$ with $\cW(n) \subseteq \cV(n)$.
Since $\cD$ is dense in $\cH$ and expressions of the form $\varphi_k(f_k) \cdots \varphi_1(f_1)\Omega$ are continuous in the $f_j$, $\cW$ is also dense.
But this means that $\cW(n)$ is dense in $\cV(n)$, and since $\cV(n)$ was assumed finite-dimensional we must have $\cW(n)=\cV(n)$.
We conclude that $\cV \subset \cD$ as claimed.

For $v \in \cV$, expressions of the form $\varphi_k(f_k) \cdots \varphi_1(f_1)v$ are jointly continuous in the $f_j$ by the same argument given above in the case $v = \Omega$.
The Wightman fields that correspond to unitary M\"obius vertex algebras satisfy an additional, apparently stronger, requirement that we now explain.
For $N \in \bbR_{\ge 0}$, the $N$-Sobolev norm on $C^\infty(S^1)$ is given by
\begin{equation}\label{eqn: sobolev norm}
\norm{f}_N = \left(\sum_{n \in \bbZ} \big|\hat f(n)\big|^2(1+n^2)^N\right)^{1/2}.
\end{equation}
The Hilbert space completion of $C^\infty(S^1)$ under this norm, denoted $H^N(S^1)$, can be identified with the subspace of $L^2(S^1)$ consisting of functions with finite $N$-Sobolev norm.
The topology on $C^\infty(S^1)$ is induced by the norms $\norm{\cdot}_N$, and a linear map from $C^\infty(S^1)$ to a Banach space $X$ is continuous precisely when it extends to a bounded linear map on some $H^N(S^1)$.
Thus the argument above for $C^\infty(S^1)$ shows that the expressions $\varphi_1(f_1) \cdots \varphi_k(f_k)v$
extend to (jointly) continuous multilinear maps $H^N(S^1)^k \to \cH$ for some $N$
that depends on $v$.

We consider the following additional property, asking that $N$ can be chosen independent of $v$:
\begin{enumerate}[{(W}1{)}]
\setcounter{enumi}{\value{wightman}}
\item \label{itm: W unif bdd order} \textbf{Uniformly bounded order}: For every collection $\varphi_1, \ldots, \varphi_k \in \cF$, there is a positive number $N$ such that for every $v \in \cV$ the map $(f_1, \ldots, f_k) \mapsto \varphi_k(f_k) \cdots \varphi_1(f_1)v$ extends continuously to $H^N(S^1)^k$.
\setcounter{wightman}{\value{enumi}}
\end{enumerate}

\begin{rem}\label{rm:conjecture}
We do not know if the condition of uniformly bounded order follows from the other axioms.
Since we will show that M\"obius-covariant Wightman field theories with uniformly bounded order correspond exactly to unitary M\"obius vertex algebras, the uniformly bounded order condition is automatic if and only if every M\"obius-covariant Wightman field theory corresponds to a vertex algebra.
The axioms of a vertex algebra essentially require the existence of operator product expansions (OPE),
while from the Wightman point of view a form of OPEs has been proved under additional analytic conditions (see e.g.\! \cite{Bostelmann05OPE}).
On the other hand, if we assume uniformly bounded order, then
we prove that the fields generate a M\"obius vertex algebra, and 
OPE $A(z)B(w) = \sum_{j=0}^{N-1} (i_{w,z}\frac1{(z-w)^{j+1}}c^j(w) + :A(z)B(w):)$ as formal series follow \cite[Theorem 2.3]{Kac98},
and this equality holds as distributions.
\end{rem}

\begin{rem}
The definition of uniformly bounded order given here implicitly assumes that $\cV \subset \cD$, which will usually only occur when the weight spaces $\cV(n)$ are finite-dimensional.
If the $\cV(n)$ are infinite-dimensional, then the condition (W\ref{itm: W unif bdd order}) should be modified to only require that the vectors $v$ lie in $\cV \cap \cD$.
\end{rem}

\subsection{Adjointable fields from unitary M\"obius vertex algebras}\label{sec:fieldsfromvertex}
Let $\cV$ be a unitary M\"obius vertex algebra.
In this section, we show that the formal distributions $Y(v,z)$ correspond to adjointable operator-valued distributions on $S^1$, and that these distributions can be repeatedly applied to the vacuum.
This is a first step in constructing a Wightman field theory from $\cV$.
Other conditions of Wightman fields, such as M\"obius covariance
and locality, will be studied in Section \ref{sec:vertexwightman}.

Recall that for $v \in \cV$ we have a formal series $Y(v,z) = \sum_n v_{(n)}z^{-n-1}$ with $v_{(n)} \in \mathrm{End}(\cV)$.
For $f \in C^\infty(S^1)$, we wish to define an unbounded operator $Y^0(v,f)$ on the Hilbert space completion $\cH$ of $\cV$ by the formula
\begin{equation}\label{eqn: smeared field Y 0 def}
Y^0(v,f)u = \sum_{n \in \bbZ} \hat f(n) v_{n}u, \quad u \in \cV.
\end{equation}
Here, $\hat f(n)$ are the coefficients of the Fourier series of $f$, and the degree-shifted mode $v_n$ is given by $v_n = v_{(n+d-1)}$ when $v$ is homogeneous with conformal dimension $d$, and otherwise extended linearly.
Formally, the map $Y^0(v,f)$ corresponds to the smeared field $\oint_{S^1} Y(z^{L_0}v,z)f(z)\,ds$, where $ds$ is normalized arclength measure.

We will see shortly that with this formula $Y^0(v,f)$ indeed maps $\cV$ into the Hilbert space completion $\cH$ of $\cV$, for any smooth $f$. 
The key estimate is given by the following lemma.

\begin{lem}\label{lem:polynbounds}
Let $\cV$ be a unitary M\"obius%
\footnote{For this lemma, the statement and proof go through equally well with $\cV$ a $\bbN$-graded vertex algebra 
and $u^\prime$ in the restricted dual $\bigoplus \cV(n)^*$.
}
vertex algebra. 
Then for all $v^1, \ldots, v^k,u,u' \in \cV$, there exists a polynomial $p$ such that 
$$
\abs{\langle v^1_{m_1} v^2_{m_2} \cdots v^k_{m_k}u, u^\prime \rangle} \le \abs{p(m_1, \ldots, m_k)}
$$
for all $(m_1, \ldots, m_k) \in \bbZ^k$.
The polynomial depends on the vectors $v^j,u,$ and $u^\prime$, but the degree of $p$ may be bounded independent of $u$ and $u'$.
\end{lem}
\begin{proof}
We proceed by induction on $k$.
The case $k=0$ is vacuous and the case $k=1$ is immediate as $\ip{v_m u,u'}$ is nonzero for only finitely many $m$.
Now suppose the result has been proven for a value $k-1$, with $k \ge 2$. We assume without loss of generality that each $v^j$ is homogeneous with conformal dimension $d_j$, and similarly for $u$ and $u'$.

First, observe that when $m_1 < -d_{u'}$ we have $ \langle v^1_{m_1} v^2_{m_2} \cdots v^k_{m_k} u, u' \rangle = 0$ by unitarity,
and thus we need only consider $m_1 \ge -d_{u'}$. We have
\begin{equation}\label{eqn: v1 moving to the right}
\langle v^1_{m_1} v^2_{m_2} \cdots v^{k}_{m_{k}}u, u^\prime \rangle
=
\langle  \big(v^2_{m_2} \cdots v^{k}_{m_{k}}\big)v^1_{m_1}u, u^\prime \rangle
+
\langle [ v^1_{m_1}, v^2_{m_2} \cdots v^{k}_{m_{k}}] u, u^\prime \rangle
.
\end{equation}

We show separately that the two terms on the right-hand side are bounded by polynomials in the $m_j$ (with degrees controlled by the $v^j$).
We begin with the first of these two terms.
Observe that when $m_1 > d_u$ we have $\langle  v^2_{m_2} \cdots v^{k}_{m_{k}}v^1_{m_1}u, u^\prime \rangle = 0$, and thus we are left with finitely many values of $m_1$ to consider.
For each fixed value of $m_1$ the inductive hypothesis asserts that $\big| \langle  v^2_{m_2} \cdots v^{k}_{m_{k}}\big(v^1_{m_1}u\big), u^\prime \rangle \big|$ is bounded by a suitable polynomial in $m_2, \ldots, m_k$, which completes the argument for the first term.

For the second term on the right-hand side of \eqref{eqn: v1 moving to the right}, we write
\[
\langle [ v^1_{m_1}, v^2_{m_2} \cdots v^{k}_{m_{k}}]u,u'\rangle = \sum_{\ell=2}^k \langle v_{m_2}^2 \cdots [v_{m_1}^1,v_{m_\ell}^\ell] \cdots v_{m_k}^k u, u' \rangle.
\]
As the sum on the right-hand side is finite, it suffices to bound each term separately by a suitable polynomial.
We will now apply the Borcherds commutator formula (see e.g.\! \cite[Section 4.8]{Kac98}), which is given in terms of degree-shifted modes by
\[
 [a_{m},b_{n}] = \sum_{s=0}^{d_a+d_b-1} \binom{m+d_a-1}{s} \big(a_{s-d_a+1} b\big)_{m+n}
\]
when $a$ and $b$ are homogeneous (with conformal dimensions $d_a$ and $d_b$, respectively).
The commutator formula yields
\begin{align} \label{eqn: expand commutator formula}
\langle v_{m_2}^2 \cdots &[v_{m_1}^1,v_{m_\ell}^\ell] \cdots v_{m_k}^k u, u' \rangle 
=\\
&=\sum_{s=0}^{d_1+d_\ell-1} \binom{m_1+d_1-1}{s}  
\langle v_{m_2}^2 \cdots \big(v^1_{s-d_1+1} v^\ell\big)_{m_1+m_\ell} \cdots v_{m_k}^k u, u' \rangle. \nonumber
\end{align}
By the inductive hypothesis, for each $s$ we may bound
\begin{align*}
\big| \langle v_{m_2}^2 \cdots \big(v^1_{s-d_1+1} v^\ell\big)_{m_1+m_\ell} \cdots v_{m_k}^k u, u' \rangle \big| &\le \abs{ \tilde p(m_1, \ldots, m_1+m_\ell, \ldots, m_k)}\\
 &\le \abs{p(m_1, \ldots, m_k)}
\end{align*}
for a polynomial $p$ in the $m_j$.
As $s$ takes only finitely many values we may choose a single $p$ that works for all $s$.
Note that the degree of $\tilde p$, and thus of $p$, may be bounded independent of $u$ and $u'$, by the inductive hypothesis.

Taking absolute values in \eqref{eqn: expand commutator formula}, we now have
\[
\big| \langle v_{m_2}^2 \cdots [v_{m_1}^1,v_{m_\ell}^\ell] \cdots v_{m_k}^k u, u' \rangle \big| \le \abs{p(m_1, \ldots, m_k)}  \sum_{s=0}^{d_1+d_\ell-1} \binom{m_1+d_1-1}{s}.
\]
Observe that $\binom{m_1+d_1-1}{s}$ is a polynomial in $m_1$ of degree $s$, and thus the above sum of binomial coefficients is a polynomial in $m_1$ of degree at most $d_1+d_\ell-1$.
This shows that $\big| \langle v_{m_2}^2 \cdots [v_{m_1}^1,v_{m_\ell}^\ell] \cdots v_{m_k}^k u, u' \rangle \big|$ is bounded by a suitable polynomial, which completes the proof.
\end{proof}

If $f$ is a Laurent polynomial (i.e. if $f$ has only finitely many nonzero Fourier coefficients), then the map $Y^0(v,f)$ defined in \eqref{eqn: smeared field Y 0 def} is a finite linear combination of $v_{(n)}$'s and thus $Y^0(v,f)$ defines a map $\cV \to \cV$.
We have the following estimate for these maps.

\begin{lem}\label{lem: finite energy norm estimate}
Let $\cV$ be a unitary M\"obius vertex algebra.
Let $v^1, \ldots, v^k, u \in \cV$, and let $f_1, \ldots, f_k \in \bbC[z^{\pm 1}]$ be Laurent polynomials.
Then there is a positive number $N$, which depends only on the $v^j$, such that
\[
\norm{Y^0(v^k,f_k) \cdots Y^0(v^1, f_1)u} \le C \norm{f_1}_N \cdots \norm{f_k}_N.
\]
The constant $C$ depends on the $v^j$ and $u$.
\end{lem}
\begin{proof}
Recall that $\norm{f}_N$ refers to the Sobolev norm \eqref{eqn: sobolev norm}.
The definition of a unitary M\"obius vertex algebra implies that
\[
\ip{v_n u, u'} = \ip{u, \big(e^{L_1}(-1)^{L_0}\Theta v\big)_{-n} u'}
\]
for all $u,u' \in \cV$ (see e.g.\! the proof of \cite[Theorem 5.16]{CKLW18}).
Thus, setting $\tilde v^j = e^{L_1}(-1)^{L_0}\Theta v^j$ we have
\begin{align*}
\big\|Y^0(v^k, f_k) &\cdots Y^0(v^1,f^1)u\big\|^2
=\\
&= \sum_{\mathclap{\substack{n_{1},\ldots,n_{k}\\
                              m_{1},\ldots,m_{k}}}}
 \hat{f}_{1}(n_{1})\overline{\hat{f}_{1}(m_{1})} \cdots \hat{f}_{k}(n_{k})\overline{\hat{f}_{k}(m_{k})} \,\, \langle \tilde v^{1}_{-m_{1}} \cdots \tilde v^{k}_{-m_{k}} v^{k}_{n_{k}} \cdots v^{1}_{n_{1}}u, \; u \rangle.
\end{align*}
By Lemma \ref{lem:polynbounds} we may choose a polynomial $p$ such that
\[
\big|\langle \tilde v^{1}_{-m_{1}} \cdots \tilde v^{k}_{-m_{k}} v^{k}_{n_{k}} \cdots v^{1}_{n_{1}}u, \; u \rangle \big|
\le 
\abs{p(m_1, \ldots, m_k, n_1, \ldots, n_k)},
\]
for all $n_j$ and $m_j$, with the degree of $p$ bounded independent of $u$.
Choose $M$ sufficiently large that
\[
\abs{p(m_1, \ldots, m_k, n_1, \ldots, n_k)} \le C\prod_{j=1}^k (1+m_j^2)^{M} (1+n_j^2)^{M}
\]
for some constant $C$. Note that we may choose $M$ independent of $u$ (although the constant $C$ does depend on $u$).

We then have
\begin{align*}
\big\|Y^0(v^k, f_k) \cdots Y^0(v^1,f^1)u\big\|^2 &\le 
C \sum_{\substack{n_{1},\ldots,n_{k}\\
                              m_{1},\ldots,m_{k}}} 
\prod_{j=1}^k \big| \hat{f}_{j}(n_{j})\hat{f}_{j}(m_{j}) \big| (1+n_j^2)^{M} (1+m_j^2)^{M}\\
&= C \prod_{j=1}^k \left(\sum_{n} \big|\hat f_j(n) \big| (1+n^2)^{M}  \right)^2\\
&\le C' \norm{f_1}^2_{2M+1} \cdots \norm{f_k}_{2M+1}^2,
\end{align*}
where in the last inequality we apply the Cauchy-Schwarz inequality to the expression
$
\big|\hat f_j(n)\big|(1+n^2)^{M+1/2} \cdot (1+n^2)^{-1/2}.
$
Note that there is no issue of convergence, as the $f_j$ are Laurent polynomials and all sums are finite.
This establishes the desired estimate with $N=2M+1$.
\end{proof}

In particular we have $\norm{Y^0(v,f)u} \le C \norm{f}_N$  when $f$ is a Laurent polynomial.
Thus, if $f \in C^\infty(S^1)$ and $f_k$ is a sequence of Laurent polynomials that converges to $f$ in $C^\infty(S^1)$, we have 
\[
\lim_{k \to \infty} Y^0(v,f_k)u = \sum_{n \in \bbZ} \hat f(n) v_n u,
\]
with convergence in $\cH$.
We may therefore define $Y^0(v,f):\cV \to \cH$ as in \eqref{eqn: smeared field Y 0 def} for arbitrary $f \in C^\infty(S^1)$.
As noted in the proof of Lemma \ref{lem: finite energy norm estimate}, the definition of a unitary M\"obius vertex algebra implies that
\[
\ip{v_n u, u'} = \ip{u, \big(e^{L_1}(-1)^{L_0}\Theta v\big)_{-n} u'}
\]
for all $u,u' \in \cV$.
Extending by continuity, we have
\[
\langle Y^0(v,f)u, u' \rangle = \langle u, Y^0(e^{L_1}(-1)^{L_0}\Theta v, \overline{f})u' \rangle
\]
for $f \in C^\infty(S^1)$, and thus we have an extension of unbounded operators
\[
Y^0(e^{L_1}(-1)^{L_0}\Theta v, \overline{f}) \subset Y^0(v,f)^* .
\]
In particular $Y^0(v,f)^*$ is densely defined, and consequently the operators $Y^0(v,f)$ are closable.
We denote the closure of $Y^0(v,f)$ by $Y(v,f)$.
Taking closures in the above extension, we obtain
\begin{equation}\label{eqn: smeared field adjoint}
Y(e^{L_1}(-1)^{L_0}\Theta v, \overline{f}) \subset Y(v,f)^* .
\end{equation}

A variant of the above argument shows that smeared fields $Y(v,f)$ may be repeatedly applied to finite-energy vectors.

\begin{prop}\label{prop: VOA has unif bdd order}
Let $\cV$ be a unitary M\"obius vertex algebra, let $v^1, \ldots, v^k \in \cV$, and let $f_1, \ldots, f_k \in C^\infty(S^1)$.
Then $\cV$ is contained in the domain of $Y(v^k,f_k) \cdots Y(v^1,f_1)$.
Moreover, there exists a number $N$ such that for every $u \in \cV$ we have
\[
\norm{Y(v^k,f_k) \cdots Y(v^1,f_1)u} \le C \norm{f_1}_N \cdots \norm{f_k}_N,
\]
for some constant $C$ (that depends on $u$).
\end{prop}
\begin{proof}
We proceed by induction on $k$.
Suppose that $k \ge 1$ and that the result holds for $k-1$ (with the base case $k=0$ being vacuous).
By the inductive hypothesis we may choose $N$ such that 
\[
\norm{Y(v^{k-1},f_{k-1}) \cdots Y(v^1,f_1)u} \le C \norm{f_1}_N \cdots \norm{f_{k-1}}_N.
\]

 By Lemma \ref{lem: finite energy norm estimate} and the following discussion, there is a continuous multilinear map $X: H^{N'}(S^1)^k \to \cH$ such that when $f_1, \ldots, f_{k-1}$ are Laurent polynomials
 \[
 X(f_1, \ldots, f_k) = Y(v^k,f_k)Y(v^{k-1},f_{k-1}) \cdots Y(v^1,f_1)u.
 \]
 Making $N$ or $N'$ larger if necessary, we assume without loss of generality that $N=N'$.

Now let $f_1, \ldots, f_k \in C^\infty(S^1)$, and for $j=1, \ldots, k-1$ choose sequences of Laurent polynomials $f_{j,n}$ such that $f_{j,n} \to f_j$ in $H^N(S^1)$.
Then, by the inductive hypothesis, we have
 \[
\lim_{n \to \infty} Y(v^{k-1}f_{k-1,n}) \cdots Y(v^1,f_{1,n})u = Y(v^{k-1},f_{k-1}) \cdots Y(v^1,f_{1})u
 \]
 in $\cH$.
 Moreover, since $X$ is continuous we have that
 \begin{align*}
 \lim_{n \to \infty} Y(v^k,f_k) Y(v^{k-1},f_{k-1,n}) \cdots Y(v^1,f_{1,n})u &= \lim_{n \to \infty} X(f_{1,n}, \ldots, f_{k-1,n}, f_k)\\ 
 &=  X(f_1, \ldots, f_k).
 \end{align*}
Since $Y(v^k,f_k)$ is closed, it follows that $Y(v^{k-1}f_{k-1}) \cdots Y(v^1,f_{1})u$ lies in the domain of $Y(v^k,f_k)$, and that
\begin{equation}\label{eqn: prod of Y equals Phi}
Y(v^k,f_k)Y(v^{k-1}f_{k-1}) \cdots Y(v^1,f_{1})u = X(f_1, \ldots, f_k).
\end{equation}
 Since $u \in \cV$ was arbitrary, we have that $\cV$ is contained in the domain of the product $Y(v^k,f_k) \cdots Y(v^1,f_1)$.
The desired estimate on $\norm{Y(v^k,f_k) \cdots Y(v^1,f_1)u}$ follows immediately from \eqref{eqn: prod of Y equals Phi} and the continuity of $X$, completing the proof.
\end{proof}

\begin{rem}
The natural analog of Proposition \ref{prop: VOA has unif bdd order} with the vertex operators $Y(v,z)$ replaced by a unitary module action $Y^M(v,z)$ also holds, with the same proof.
However we will not discuss modules for vertex algebras in this article.
\end{rem}

\begin{cor}\label{cor: VOA gives adjointable distributions}
Let $\cV$ be a unitary M\"obius vertex algebra and let $v \in \cV$.
Let 
\[
\cD = \operatorname{Span} \{ Y(v^k, f_k) \cdots Y(v^1,f_1)\Omega \,\, : \,\, v^j \in \cV, \, f_j \in C^\infty(S^1), k \in \bbZ_{\ge 0}\}.
\]
Then the assignment $\varphi_v(f) = Y(v,f)\mid_\cD$ is an adjointable operator-valued distribution on $\cD$.
If $\tilde v = e^{L_1}(-1)^{L_0}\Theta v$, then $\varphi_v^\dagger = \varphi_{\tilde v}$.
\end{cor}
\begin{proof}
The fact that $\varphi_v$ and $\varphi_{\tilde v}$ are operator-valued distributions on $\cD$ follows immediately from Proposition \ref{prop: VOA has unif bdd order}.
Adjointability and the formula for $\varphi_v^\dagger$ now follow immediately from Equation \eqref{eqn: smeared field adjoint}.
\end{proof}

Note that Corollary~\ref{cor: VOA gives adjointable distributions} applies to an arbitrary unitary M\"obius vertex algebra $\cV$, in contrast to earlier results (e.g. in \cite{CKLW18}) which rely on specialized analytic assumptions on $\cV$, such as polynomial energy bounds.

\section{Vertex operators and M\"obius-covariant Wightman fields}\label{sec:vertexwightman}

In this section, we will show that the distributions described in Corollary \ref{cor: VOA gives adjointable distributions} satisfy the axioms of Wightman field theory.
We will also establish a converse result, that the space of finite-energy vectors in a Wightman field theory with uniformly bounded order can be equipped with a structure of M\"obius vertex algebra.

One of the technical challenges in establishing the Wightman axioms is that all fields act on a common domain $\cD$ which is in turn determined collectively by all of the fields (as it is spanned by expressions of the form $\varphi_k(f_k) \cdots \varphi_1(f_1)\Omega$).
It is easier to consider properties that can be studied for each field separately, which we will do by restricting our Wightman fields to the smaller domain of finite-energy vectors $\cV$.
In this section, we will formulate a notion of ``quasi-Wightman field theory'' which corresponds to a family of operator-valued distributions with domain $\cV$ that satisfy analogs of the axioms of Wightman fields, and we will show that, under the assumption of uniformly bounded order, Wightman field theories and quasi-Wightman field theories are both equivalent to unitary M\"obius vertex algebras.

\subsection{Regularity of adjointable fields on \texorpdfstring{$\cV$}{V}}\label{sec:regularity}

As before, fix a Hilbert space $\cH$ carrying a positive-energy representation $U:\Mob \to \cU(\cH)$ of $\Mob$, and denote by $\cV = \bigoplus_{n \in \bbZ_{\ge 0}} \cV(n)$ the dense subspace of finite-energy vectors. We assume that each $\cV(n)=\mathrm{ker}(L_0-n)$ is finite-dimensional and that $\dim \cV(0) = 1$.

 As discussed above, it will be a useful technical tool to consider a variation of the Wightman axioms in which the operator-valued distributions have domain $\cV$, and not a larger invariant domain.
 Working with such theories, which we call \textbf{quasi-Wightman field theories}, involve several technical trade-offs.
For our purposes, a key benefit is that quasi-Wightman field theories are easier to construct than Wightman field theories, since almost all of the imposed conditions are phrased in terms of properties of individual fields as opposed to interactions of families of fields.

We will now study general properties of operator-valued distributions defined on the domain $\cV$.
Write $\widehat \cV = \prod_{n \in \bbZ_{\ge 0}} \cV(n)$ for the algebraic completion of $\cV$, namely, the set of all sequences $(v_1, v_2,\cdots )$ with $v_n \in \cV(n)$, equipped with the product topology.
We have inclusions $\cV \subset \cH \subset \widehat \cV$ as linear spaces.
The inner product on $\cV$ gives an identification of $\widehat \cV$ with $\overline{\cV}^*$, the algebraic dual space of the complex conjugate of $\cV$.
We will generally use $\ip{ \, \cdot \, , \, \cdot \,}$ to denote both the inner product on $\cH$ and the induced sesquilinear pairing of $\cV$ with $\widehat \cV$, adding subscripts to clarify when necessary.
The topology on $\widehat \cV$ is induced by the family of linear functionals $\ip{ \, \cdot \, , u}$, for $u \in \cV$.

\begin{lem}\label{lem: extension of Wightman fields}
Let $\varphi$ be an operator-valued distribution with domain $\cV$.
Then the following are equivalent:
\begin{enumerate}
\item $\varphi$ is adjointable.\label{key:ovdist}
\item The map $\varphi:C^\infty(S^1) \times \cV \to \cH$ extends to a continuous map
$C^\infty(S^1) \times \cH \to \widehat \cV$.
\label{key:bicontinuous}
\end{enumerate}
\end{lem}
\begin{proof}

Note that joint and separate continuity of maps $C^\infty(S^1) \times \cH \to \widehat \cV$ are equivalent by the Banach-Steinhaus theorem
for Fr\'echet spaces (see e.g.\! \cite[Theorem 2.17]{Rudin91}).

First assume that $\varphi$ is an adjointable field.
For each $f \in C^\infty(S^1)$ we extend $\varphi(f)$ to a map $\cH \to \widehat \cV$ 
as follows.
For $\xi \in \cH$, we define $\varphi(f)\xi \in \widehat \cV \cong \overline{\cV}^*$ by requiring
\[
\ip{\varphi(f)\xi,v}_{\widehat \cV, \cV} = \ip{\xi, \varphi^\dagger(\overline{f})v}_{\cH}.
\]
The latter expression is evidently continuous in $\xi$ and $f$ separately, hence jointly.

Conversely, assume that $\varphi$ extends to a continuous map $C^\infty(S^1) \times \cH \to \widehat \cV$.
We use the symbol $\varphi(f)$ to denote both the map $\cV \to \cH$ as well as the extension $\cH \to \widehat \cV$.
For each $f \in C^\infty(S^1)$ and $v \in \cV$ the map $\cH \to \bbC$ given by $\xi \mapsto \ip{\varphi(f)\xi,v}$ is continuous by assumption.
Restricting to $\xi \in \cV$ we have $v \in \dom(\varphi(f)^*)$ for every $f$.
Define $\varphi^\dagger:C^\infty(S^1) \times \cV \to \cH$ by the formula $\varphi^\dagger(f)v = \varphi(\overline{f})^*v$.
By construction we have 
\[
\ip{\xi, \varphi^\dagger(\overline{f})v}_{\cH} = \ip{\varphi(\overline{f})\xi,v}_{\widehat \cV,\cV},
\]
for $\xi,v \in \cV$, and by continuity this extends to all $\xi \in \cH$.
By assumption the right-hand expression is continuous in $f$.
Thus $f\mapsto \varphi^\dagger(f)v$ is continuous as a map into $\cH$ equipped with the weak topology.
This then implies that this map is continuous into $\cH$ with the norm topology by the closed graph theorem (see e.g.\! \cite[Theorem 5.1]{Rudin91}),
and thus $\varphi^\dagger$ is an operator-valued distribution.
\end{proof}

As a consequence of Lemma~\ref{lem: extension of Wightman fields}, when $\varphi$ is an adjointable field with domain $\cV$
we will implicitly extend $\varphi(f)$ to a continuous map $\varphi(f):\cH \to \widehat \cV$.
We will generally consider compound expressions which define vectors in $\widehat \cV$ or some subspace.
For example, if $\varphi$ and $\psi$ are adjointable fields and $u \in \cV$, then we have $\varphi(f)\psi(g)u \in \widehat \cV$.

As described in Section \ref{sec: two formulations}, Wightman fields arising from vertex operators are not only adjointable; they also satisfy a certain uniformity condition on the orders of the distributions. We describe this condition below in the context of a single operator-valued distribution with domain $\cV$.
This is an analog of the condition (W\ref{itm: W unif bdd order}) of uniformly bounded order for a Wightman field theory, which is phrased for a collection of Wightman fields acting on an invariant domain.
This condition will eventually play a crucial role in comparing locality of fields in the Wightman sense with locality in the sense of vertex operators.

\begin{defn}\label{defn: uniformly bounded order}
An adjointable operator-valued distribution $\varphi$ with domain $\cV$ is said to have \textbf{uniformly bounded order}
if there is a positive number $N$ such that the assignment $f \mapsto \varphi(f)v$ extends to a continuous map $H^N(S^1) \to \cH$ for every $v \in \cV$.
\end{defn}

Note that for each $v$ we are guaranteed that the map $f \mapsto \varphi(f)v$ extends continuously to $H^N(S^1)$ for some $N$ by the definition of the topology on $C^\infty(S^1)$, but in Definition \ref{defn: uniformly bounded order} we insist on choosing $N$ independent of $u$.

We have the following analog of Lemma~\ref{lem: extension of Wightman fields}.
\begin{lem}\label{lem: uniformly bounded order of extension}
Let $\varphi$ be an adjointable field with domain $\cV$.
Then the following are equivalent:
\begin{enumerate}
\item $\varphi^\dagger$ has uniformly bounded order.
\item For some $N>0$ the map $\varphi:C^\infty(S^1) \times \cH \to \widehat \cV$ extends to a continuous map $H^N(S^1) \times \cH \to \widehat \cV$.
\end{enumerate}
\end{lem}
\begin{proof}
This can be proved in a parallel way as in Lemma~\ref{lem: extension of Wightman fields},
replacing the Fr\'echet space $C^\infty(S^1)$ by the Hilbert space $H^N(S^1)$.
\end{proof}

\subsection{M\"obius covariance for adjointable fields}
As above, let $\cH$ be a Hilbert space carrying a positive-energy representation $U:\Mob \to \cU(\cH)$, and let $\cV$ be the subspace of finite-energy vectors, with finite-dimensional weight spaces $\cV(n) = \mathrm{ker}(L_0 - n)$ and $\dim \cV(0) = 1$.
As described in Section \ref{sec: PERs}, $U$ is the exponentiation of a positive-energy representation of $\Lie(\Mob)$.
Given a complexified vector field $g \frac{d}{d\theta} \in \Lie(\Mob)_\bbC$, we write $\pi(g \tfrac{d}{d\theta}) \in \spann_\bbC \{L_{-1},L_0,L_1\}$ for the corresponding operator, which we may regard as a map of $\cV$ into $\cV$ or as a map of $\widehat \cV$ into $\widehat \cV$.
Thus, if $\varphi$ is an adjointable field, we may consider expressions $\pi(g \tfrac{d}{d \theta}) \varphi(f)u \in \widehat \cV$ and $\varphi(f)\pi(g \tfrac{d}{d\theta})u \in \cH$.
This allows us to state transformation rules for $\varphi$ under the action of $\Lie(\Mob)$ in terms of familiar commutation relations.

\begin{rem}\label{rem: gddtheta domain}
If $g\tfrac{d}{d\theta} \in \Lie(\Mob)$, then $i\pi(g \tfrac{d}{d\theta})$ is a linear combination of $L_1$, $L_0$ and $L_{-1}$, and so it
is essentially skew-adjoint on $\cV$ by the commutator theorem \cite[Theorem X.37]{RSII} and the explicit estimates in \cite[Section 1.5]{Longo08}.
Hence the domain of its closure on $\cV$ coincides with the domain of $(\pi(g \tfrac{d}{d\theta})|_\cV)^*$.
From this point of view, the operator $\pi(g \tfrac{d}{d\theta})$ can be seen as a map from $\cH$ to $\widehat\cV$ and the domain of the closure of $\pi(g \tfrac{d}{d\theta})\!\mid_\cV$ consists of the vectors $\xi \in \cH$ such that $\pi(g \tfrac{d}{d\theta})\xi \in \cH$.
\end{rem}

An adjointable field $\varphi$ on $\cV$ is called \textbf{M\"obius-covariant} (or \textbf{quasi-primary}) \textbf{with degree $d$}, \textbf{in the Wightman sense}
if, for every $f \in C^\infty(S^1)$, $u \in \cV$, and $\gamma \in \Mob$, we have\footnote{While \textit{a priori} $\varphi(f)U(\gamma)^*u \in \widehat \cV$ we require that it is, in fact, in $\cH$} $\varphi(f)U(\gamma)^*u \in \cH$,
and moreover $\varphi$ satisfies the covariance condition
\[
U(\gamma)\varphi(f)U(\gamma)^*u = \varphi(\beta_d(\gamma)f)u,
\]
where $\beta_d$ is defined in \eqref{eq:beta}.
Similarly, an adjointable field $\varphi$ on $\cV$ is called \textbf{M\"obius-covariant} (or \textbf{quasi-primary})
\textbf{with degree $d$ in the vertex algebra sense}
if it satisfies the infinitesimal covariance condition with degree $d$:
\begin{align}\label{eq:infcomm}
[\pi(g \tfrac{d}{d \theta}), \varphi(f)]u = \varphi((d-1)\tfrac{dg}{d\theta} f - g \tfrac{d f}{d\theta})u
\end{align}
for all $g \tfrac{d}{d \theta} \in \Lie(\Mob)_\bbC$, $u \in \cV$, and $f \in C^\infty(S^1)$.
For a quasi-primary field $\varphi$ with degree $d$, we have a formal series expansion
$\hat \varphi(z) = \sum_n \varphi(e_n) z^{-n-d}$, where $e_n(e^{i\theta}) = e^{in\theta}$.
By taking $g(e^{i\theta}) = -ie^{ik\theta}$, we see that M\"obius covariance in the vertex algebra sense is equivalent to the familiar condition
\begin{align}\label{eq:quasi-primary}
 [L_{k}, \hat \varphi(z)] = (z^{k+1} \tfrac{d}{dz} + (k+1)z^k d) \hat \varphi(z), \qquad k=-1,0,1.
\end{align}

\begin{lem}\label{lem: Mobius covariance conditions}
Let $\varphi$ be an adjointable operator-valued distribution with domain $\cV$ and let $d$ be a non-negative integer.
Then $\varphi$ is M\"obius-covariant with degree $d$ in the Wightman sense 
if and only if it is so in the vertex algebra sense.
\end{lem}
\begin{proof}
First assume that $\varphi$ satisfies the infinitesimal covariance condition \eqref{eq:infcomm}.
Let $g \frac{d}{d \theta} \in \Lie(\Mob)$ be a real vector field, and let $\gamma_t = \exp(t g \frac{d}{d \theta}) \in \Mob$ be the corresponding one-parameter group.
Note that the exponential map from $\Lie(\Mob)$ to $\Mob$ is surjective, so it suffices to prove  covariance for transformations of the form $\gamma_t$.
The operator $\pi(g \frac{d}{d\theta})$ is skew-adjoint (we work with the closure of the operator defined on finite-energy vectors) and we have $U(\gamma_t) = e^{t \pi(g \frac{d}{d\theta})}$
(see \cite[Section 1.5]{Longo08}).
Fix $u \in \cV$, and let 
$$
u_1(t) = \varphi(f)U(\gamma_{-t})u,
\quad \mbox{ and } \quad
u_2(t) = U(\gamma_{-t})\varphi(\beta_d(\gamma_{t})f)u.
$$
Note that $u_1$ naturally takes values in $\widehat \cV$ and $u_2$ naturally takes values in $\cH$.
We would now like to show that $u_1 = u_2$.
As both $u_1(t)$ and $u_2(t)$ are continuous in $f$ (as maps into $\widehat \cV$), we may assume without loss of generality that $f$ is a Laurent polynomial.

We first argue that both $u_1$ and $u_2$ extend to holomorphic functions on a neighborhood of $\bbR$.
Indeed, by \cite[Proposition A.2.9]{WeinerThesis}, the function $U(\gamma_{-t})u = U(\gamma_{-t_0})U(\gamma_{-(t-t_0)})u$ extends to a holomorphic function into $\cH$ defined on a neighborhood of an arbitrary $t_0 \in \bbR$.
Composing with the continuous linear map $\varphi(f): \cH \to \widehat \cV$ yields a $\widehat \cV$-valued holomorphic extension of $u_1$
(see \cite[Definition 3.30]{Rudin91}).

We now consider $u_2$.
As each $\gamma_t$ maps $S^1$ onto itself, the map $t \mapsto \gamma_t$ has a holomorphic extension to a neighborhood of $\bbR$, taking values in M\"obius transformations that map $S^1$ into $\bbC^\times$.
Since $f$ is a Laurent polynomial, the map $(t,z) \mapsto (\beta_d(\gamma_t)f)(z)$ has an extension to a neighborhood of $\bbR \times S^1$
which is (separately, and therefore jointly) holomorphic.
Thus the function $t \mapsto \beta_d(\gamma_t)f$ extends to a holomorphic function from a neighborhood of $\bbR$ into $C^\infty(S^1)$,
since the $z$-derivatives of $(\beta_d(\gamma_t)f)(z)$ are uniformly continuous in a compact neighborhood of $S^1$.
Composing with the continuous linear map $\varphi$ shows that $t \mapsto \varphi(\beta_d(\gamma_t)f)u$ has a holomorphic extension as a map into $\cH$.
Hence,
\[
\ip{u_2(t), v} = \ip{\varphi(\beta_d(\gamma_t)f)u, U(\gamma_{-t})^*v}
\]
has a holomorphic extension for every $v \in \cV$.
We conclude that $u_2$ extends to a holomorphic map from a neighborhood of $\bbR$ into $\widehat \cV$.

As both $u_1$ and $u_2$ are analytic, it suffices to show that their Taylor series centered at $t=0$ coincide.
For $u_1$, it is straightforward to compute the derivatives 
\[
\frac{d}{dt} \ip{u_1(t),v} = 
\frac{d}{dt} \ip{U(\gamma_{-t})u,\varphi^{\dagger}(\overline{f})v} = 
\ip{\varphi(f) U(\gamma_{-t})\pi(-g \tfrac{d}{d\theta}) u, v}.
\]
Repeatedly differentiating (while noting that $\pi(-g\tfrac{d}{d\theta})$ maps $\cV$ into itself) we obtain 
\begin{equation}\label{eqn: u1 derivatives}
u_1^{(n)}(0) = \varphi(f) \pi(-g \tfrac{d}{d\theta})^n u.
\end{equation}

We now turn our attention to $u_2$.
Recall from Remark \ref{rem: tensor field} that $\beta_d(\gamma_t)f$ is given by the pushforward
\[
{\gamma_t}_* \left(f \tfrac{d}{d\theta}^{\otimes d-1}\right) = \beta_d(\gamma_t)f \tfrac{d}{d\theta}^{\otimes d-1},
\]
where in the case $d=0$ we formally set $\tfrac{d}{d\theta}^{\otimes -1} := d\theta$.
We thus have
\begin{align}\label{eqn: ddt betad gammat}
\frac{d}{dt} \, \beta_d(\gamma_t)f \tfrac{d}{d\theta}^{\otimes d-1} &=
\frac{d}{dt} {\gamma_t}_* \left(f \tfrac{d}{d\theta}^{\otimes d-1}\right)
=
-\left.\frac{d}{ds}\right|_{s=0} {\gamma_{-s}}_* \left({\gamma_t}_* \left( f \tfrac{d}{d\theta}^{\otimes d-1}\right)\right)\\
& =
-\cL_{g \tfrac{d}{d\theta}} \left( \beta_d(\gamma_t)f\tfrac{d}{d\theta}^{\otimes d-1}\right) \nonumber\\
&=\left((d-1)\tfrac{dg}{d\theta} \beta_d(\gamma_t)f - g \tfrac{d}{d\theta}\!\left[\beta_d(\gamma_t)f\right]\right)\tfrac{d}{d\theta}^{\otimes d-1} , \nonumber
\end{align}
where $\cL$ is the Lie derivative of tensor fields (see \cite[Chapter 12]{Lee13SmoothManifolds}).
As $f \mapsto \varphi(f)u$ is continuous, we can plug in the formula for $\tfrac{d}{dt} \beta_d(\gamma_t)f$ from \eqref{eqn: ddt betad gammat} and compute
\begin{align*}
\frac{d}{dt} \ip{u_2(t), v} &= \frac{d}{dt} \ip{ \varphi(\beta_d(\gamma_t)f)u, e^{t \pi(g \tfrac{d}{d\theta})}v}\\
&= \ip{ \varphi\big((d-1)\tfrac{dg}{d\theta} \beta_d(\gamma_t)f - g \tfrac{d}{d\theta}\!\left[\beta_d(\gamma_t)f\right]\big)u, e^{t \pi(g \tfrac{d}{d\theta})}v}\\
&\phantom{=}+ \ip{ \varphi(\beta_d(\gamma_t)f)u, \pi(g \tfrac{d}{d\theta})e^{t \pi(g \tfrac{d}{d\theta})}v}.
\end{align*}

By the infinitesimal commutation relation \eqref{eq:infcomm}, $\pi(g \tfrac{d}{d\theta})\varphi(\beta_d(\gamma_t)f)$ is a linear combination of vectors in $\cH$, and thus lies in $\cH$ as well.
By Remark \ref{rem: gddtheta domain}, 
$\varphi(\beta_d(\gamma_t)f)u$ lies in the domain of the skew-adjoint unbounded operator $\pi(g \tfrac{d}{d\theta})$.
We can now continue the above calculation:
\begin{align*}
\frac{d}{dt} \ip{u_2(t), v} &=
\ip{[\pi(g\tfrac{d}{d\theta}), \varphi(\beta_d(\gamma_t)f)]u, e^{t \pi(g \tfrac{d}{d\theta})}v}
-
\ip{\pi(g\tfrac{d}{d\theta}) \varphi(\beta_d(\gamma_t)f)u, e^{t \pi(g \tfrac{d}{d\theta})}v}\\
&= \ip{\varphi(\beta_d(\gamma_t)f)\pi(-g\tfrac{d}{d\theta})u, e^{t \pi(g \tfrac{d}{d\theta})}v}\\
&= \ip{e^{-t \pi(g \tfrac{d}{d\theta})}\varphi(\beta_d(\gamma_t)f)\pi(-g\tfrac{d}{d\theta})u, v}.
\end{align*}
We now repeatedly apply this computation to conclude that
\[
{\frac{d}{dt}}^n \ip{u_2(t), v} = \ip{e^{-t \pi(g \tfrac{d}{d\theta})}\varphi(\beta_d(\gamma_t)f)\pi(-g\tfrac{d}{d\theta})^n u, v}.
\]
Combining this with \eqref{eqn: u1 derivatives}, we obtain
\[
u_2^{(n)}(0) = \varphi(f) \pi(-g \tfrac{d}{d\theta})^n u = u_1^{(n)}(0)
\]
for all $n\ge0$.
As $u_1$ and $u_2$ are analytic, we conclude that $u_1 = u_2$.
The map $u_2$ takes values in $\cH$, so it follows that $u_1$ does as well.
Hence $\varphi(f)U(\gamma_{-t})u \in \cH$, and multiplying $u_1$ and $u_2$ by $U(\gamma_t)$ we get
\[
U(\gamma_t)\varphi(f)U(\gamma_{-t})u = \varphi(\beta_d(\gamma_t)f)u,
\]
completing one direction of the proof.

Now assume conversely that $\varphi$ satisfies the Wightman version of M\"obius covariance.
Without loss of generality we take $f$ to be a Laurent polynomial as above, as the infinitesimal covariance condition is continuous in $f$.
It suffices to consider $g \frac{d}{d\theta} \in \Lie(\Mob)$, and we take $\gamma_t$ as above.
The functions $u_1$ and $u_2$, defined as above, are analytic, and this time by assumption we have $u_1 = u_2$.
Evaluating the derivatives as before, we have
\[
u_1^\prime(0) = -\varphi(f)\pi(g \tfrac{d}{d\theta})u
\] 
and 
\[
u_2^\prime(0) = -\pi(g \tfrac{d}{d\theta})\varphi(f)u + \varphi((d-1)\tfrac{dg}{d\theta} f - g \tfrac{d f}{d\theta})u.
\]
The infinitesimal covariance condition follows immediately from the equality of these two quantities.
\end{proof}

\subsection{Locality for adjointable fields}
We continue with our Hilbert space $\cH$ carrying a positive-energy representation $U$ of $\Mob$.
We have finite-energy vectors $\cV = \bigoplus_{n =0}^\infty \cV(n)$, with each $\cV(n)$ finite-dimensional and $\dim \cV(0) = 1$.
We now turn to discussing the notion of locality for a pair of adjointable fields $\varphi$ and $\psi$ with domain $\cV$.
Using the extension of adjointable fields to maps $\cH \to \widehat \cV$ from Lemma~\ref{lem: extension of Wightman fields}, we may consider products like $\varphi(f)\psi(g)$ as maps $\cV \to \widehat \cV$, and we can define locality in terms of these products.

Our goal is to relate adjointable fields to point-like formal distributions as studied in the context of vertex algebras.
Let $\varphi$ be an adjointable field with domain $\cV$ that is M\"obius-covariant with degree $d$.
Let $\varphi_n = \varphi(e_n)$, where $e_n(z) = z^n$, and observe that $\varphi_n \in \End(\cV)$ by rotation covariance.
The associated formal distribution $\hat \varphi(z) \in \End(\cV)[[z^{\pm 1}]]$ is given by 
\begin{align}\label{eqn: formal distribution}
\hat \varphi(z) = \sum_{n \in \bbZ} \varphi_n z^{-n-d}.
\end{align}

Now let $\varphi$ and $\psi$ be a pair of M\"obius-covariant adjointable fields with domain $\cV$, and let $\hat \varphi$ and $\hat \psi$ be the associated formal distributions.
We say that $\varphi$ and $\psi$ are \textbf{relatively local in the vertex algebra sense} if 
for $N$ sufficiently large we have an equality of formal series $(z-w)^N [\hat \varphi(z), \hat \psi(w)] = 0$.
On the other hand, $\varphi$ and $\psi$ are called \textbf{relatively local in the Wightman sense}
if whenever $f,g \in C^\infty(S^1)$ have disjoint support we have $\varphi(f)\psi(g) = \psi(g)\varphi(f)$ as maps  $\cV \to \widehat \cV$.

The equivalence of Wightman and vertex algebra locality was established in \cite[Appendix A]{CKLW18}
under the hypothesis of \emph{polynomial energy bounds}.
It is not known whether every unitary VOA satisfies this property.
Below, we observe that the argument from \cite{CKLW18} goes through under the weaker hypothesis of uniformly bounded order (Definition \ref{defn: uniformly bounded order}), which we have shown holds in every unitary M\"obius vertex algebra (see Proposition~\ref{prop: VOA has unif bdd order}).

First, we need a technical result.
Recall that a distribution on $S^1 \times S^1$ with order $N$ is a distribution that
can be written as $\sum_\alpha \partial^\alpha f_\alpha$, where $f_\alpha$ are continuous functions on $S^1\times S^1$
and $\alpha$ are multi-indices with $|\alpha| \le N$.

\begin{lem}\label{lem: distribution on diagonal}
 Suppose $\Lambda$ is a distribution on $S^1\times S^1$ supported on the diagonal $z=w$ which has order $N$.
 Then there are distributions $c_j$ on $S^1$, for $j=0, \ldots, N$,  such that $\Lambda(z,w) = \sum_{j=0}^N \delta^{(j)}(z-w)c_j(w)$.
 In particular, $\Lambda(z,w)(z-w)^{N+1} = 0$.
\end{lem}
\begin{proof}
For a distribution $\Lambda$ on $\bbR$ supported on a point $p$ with order $N$ there are constants $c_\alpha$ such that $\Lambda = \sum_{|\alpha|\le N} c_\alpha D^\alpha \delta_p$ (see \cite[Theorem 6.25]{Rudin91}).
 The same holds for a distribution $\Lambda$ on $S^1$, because this property is determined in a neighbourhood of $p$.

 Let $f,g$ be any smooth function on $S^1$, and let $F_{f,g}(z,w) = f(z-w)g(w)$. Then the map
 $f \longmapsto \Lambda(F_{f,g})$ is a distribution on $S^1$ of order $N$.
 As this is supported at $0$, it follows that $\Lambda(F_{f,g}) = \sum_{j=0}^n c_{j,g}\delta^{(j)}(f)$,
 where $\delta = \delta_0$ is the delta-distribution supported at $0$,
 with coefficients $c_{j,g}$ that depend on $g$.
 
 We claim that the map $g\longmapsto c_{j,g}$ is a distribution on $S^1$. Indeed,
if we choose a test function $f$ such that only the $j$-th derivative does not vanish at $0$ (among $0$-th to $N$-th derivatives),
 the map $g\longmapsto \Lambda(F_{f,g}) = \sum_{j=0}^n c_{j,g}\delta^{(j)}(f) = c_{j,g}f^{(j)}(0)$
 is a distribution. That is, we found distributions $c_j(g) = c_{j,g}$.
 
For any pair of test functions $f$ and $g$ we have
 $\Lambda(F_{f,g}) = \sum_{j=0}^N c_j(g) f^{(j)}(0)$, which is to say that $\Lambda$ coincides with
 $\sum c_j(w)\delta^{(j)}(z-w)$ on such pairs.
 As $\Lambda$ is a distribution, this equality extends to any two-dimensional test function
 \cite[Theorem 39.2]{Treves67}.
\end{proof}

We can now establish the equivalence of Wightman and vertex algebra locality under the hypothesis of uniformly bounded order.

\begin{lem}\label{lem: wightman locality vs vertex locality}
Let $\varphi$ and $\psi$ be adjointable fields with domain $\cV$.
Suppose that $\varphi$, $\psi$, $\varphi^\dagger$, and $\psi^\dagger$ have uniformly bounded order.
Then $\varphi$ and $\psi$ are local in the Wightman sense if and only if they are local in the vertex algebra sense.
\end{lem}
\begin{proof}
Since the fields have uniformly bounded order, we may invoke Lemma \ref{lem: uniformly bounded order of extension} to see that the assignments $(f, g) \mapsto \varphi(f)\psi(g)u$ and  $(f, g) \mapsto \psi(g)\varphi(f)u$ give continuous bilinear maps $H^N(S^1) \times H^N(S^1) \to \widehat \cV$
for some number $N$ which does not depend on $u$.
By Lemma \ref{lem:sobolevfunctional}, for any $u,v \in \cV$ we have a continuous functional $X_{u,v}: H^M(S^1 \times S^1) \to \bbC$ determined by 
\[
X_{u,v}(f(z)g(w)) = \ip{[\varphi(f),\psi(g)]u,v}
\]
with $M = 2N+2$. This is a distribution of order $M+2$ as we noted in Appendix \ref{sec:sobolevdist}.
Furthermore, $M$ is independent of $u$ and $v$. From here, we proceed as in \cite[Proposition A.1]{CKLW18}.

First assume that $\varphi$ and $\psi$ are local in the Wightman sense. Then
the distribution $X_{u,v}$ is supported on $z=w$, and by Lemma \ref{lem: distribution on diagonal}
we infer that $(z-w)^{M+3} X_{u,v} = 0$.
Specializing $f$ and $g$ to functions of the form $z^k$ and $w^\ell$, we get the same identity at the level of formal distributions:
\[
(z-w)^{M+3} \ip{[\hat \varphi(z), \hat \psi(w)]u,v} = 0.
\]
As $u$ and $v$ were arbitrary and $M$ did not depend on these vectors, we obtain the vertex algebra locality condition
$(z-w)^{M+3} [\hat \varphi(z), \hat \psi(w)]=0$.

Conversely, suppose that $\hat \varphi$ and $\hat \psi$ are local in the sense of vertex algebras.
By \cite[Corollary 2.2]{Kac98}, we have that $X_{u,v}$ is supported on the diagonal and thus $[\varphi(f),\psi(g)]u = 0$ when $f$ and $g$ have disjoint support.
\end{proof}

We note that the property of uniformly bounded order was essential in going from Wightman locality to vertex locality in the proof of Lemma \ref{lem: wightman locality vs vertex locality}, allowing us to conclude that $M$ was independent of $u$.

\subsection{From vertex algebras to Wightman fields and back}\label{subsec: VA to W and back}

Let us prove the first of our main results, namely that
Wightman fields can be constructed from unitary M\"obius vertex algebras $\cV$.
Recall from Section \ref{sec:fieldsfromvertex} that we define closed unbounded operators $Y(u,f)$ on the Hilbert space completion $\cH$ of $\cV$, and Proposition \ref{prop: VOA has unif bdd order} says that the domain of $Y(u,f)$ contains expressions of the form $Y(v^k,f_k) \cdots Y(v^1,f_1)\Omega$.
Let $\cD = \operatorname{span}\{Y(v^k, f_k) \cdots Y(v^1,f_1)\Omega: v_j \in \cV, f_j \in C^\infty(S^1)\}$.
By Corollary \ref{cor: VOA gives adjointable distributions}, we have adjointable distributions $\varphi_u:C^\infty(S^1)\times \cD \to \cH$ given by 
\begin{equation}\label{eqn: def varphiu}
\varphi_u(f) = Y(u,f)\mid_\cD,
\end{equation}
and if $u$ is quasi-primary then $\varphi_u^\dagger = (-1)^{d_u} \varphi_{\Theta u}$.

\begin{thm}\label{thm: VA to Wightman}
Let $\cV$ be a unitary M\"obius vertex algebra with $\dim \cV(0)=1$, and let $\cH$ be the Hilbert space completion of $\cV$.
Let $\cS \subset \cV$ be a generating set for $\cV$ consisting of quasi-primary vectors, and assume that $\cS= (-1)^{L_0}\Theta \cS$.
Let $\cF = \{ \varphi_u \, : \, u \in \cS\}$, for $\varphi_u$ as in \eqref{eqn: def varphiu}.
Then $\cF$ is a M\"obius-covariant Wightman field theory on $\cH$ with uniformly bounded order.
\end{thm}
\begin{proof}
We verify that $\cF$ satisfies the axioms (W\ref{itm: W Mob covariance})-(W\ref{itm: W unif bdd order}) of a Wightman field theory, leaving the nontrivial arguments for (W\ref{itm: W Mob covariance}) and (W\ref{itm: W locality}) for the end.

By Corollary \ref{cor: VOA gives adjointable distributions}, the $\varphi_u$ are adjointable distributions and if $\varphi_u \in \cF$, then $\varphi_u^\dagger = \varphi_{(-1)^{L_0}\Theta u} \in \cF$ as well.
Hence, $\cF$ satisfies the adjoint axiom (W\ref{itm: W adjoint}).
As described in Section \ref{sec: PERs}, we may exponentiate the positive-energy representation of $\Lie(\Mob)_\bbC$ spanned by $\{L_{-1},L_0,L_1\}$ to a positive-energy representation of $\Mob$.
As (the closure of) $L_{0}$ is the generator of rotations, the spectrum condition axiom (W\ref{itm: W spec}) is satisfied by the definition of a unitary M\"obius vertex algebra. 
The vacuum axiom (W\ref{itm: W Vacuum}) is satisfied by construction, and we note that $\cD$ contains $\cV$ since $\cS$ is a generating set, and thus $\cD$ is dense.
Finite-dimensionality of the $L_0$-eigenspaces (W\ref{itm: W V(n) finite dim}) is part of the definition of M\"obius vertex algebra.
The fields $\varphi_u \in \cF$ satisfy the uniformly bounded order axiom (W\ref{itm: W unif bdd order}) by Proposition \ref{prop: VOA has unif bdd order}.
It remains to check M\"obius covariance (W\ref{itm: W Mob covariance}) and locality (W\ref{itm: W locality}).

We first consider M\"obius covariance.
From the definition of a M\"obius vertex algebra, if $u$ is quasi-primary we have $[L_k,Y(u,z)] = \left(z^{k+1} \tfrac{d}{dz} + (k+1)z^k d_u \right)Y(u,z)$ for $k=-1,0,1$, which implies that $\varphi_u$ is M\"obius-covariant in the vertex algebra sense.
By Lemma \ref{lem: Mobius covariance conditions}, it follows that if $v \in \cV$ and $\gamma \in \Mob$, we have 
$$
U(\gamma)\varphi_u(f)v = \varphi_u(\beta_d(\gamma)f) U(\gamma)v.
$$
Applying this relation repeatedly, we see that when $f_1, \ldots, f_m$ are Laurent polynomials and $u_1, \ldots, u_m \in \cS$  we have
$$
U(\gamma)\varphi_{u_1}(f_1) \cdots \varphi_{u_m}(f_m)\Omega
=
\varphi_{u_1}(\beta_{d_1}(\gamma)f_1) \cdots  \varphi_{u_m}(\beta_{d_m}(\gamma)f_m)\Omega.
$$
Since both sides are continuous in the $f_j$ by Proposition \ref{prop: VOA has unif bdd order}, the relation extends to all $f_j \in C^\infty(S^1)$, and we see that 
$U(\gamma)\varphi_u(f) U(\gamma)^* = \varphi_u(\beta_d(\gamma) f)$ as operators on $\cD$, which verifies the M\"obius covariance axiom.

We give a similar proof of the locality axiom.
Suppose that $f_1$ and $f_2$ are supported in disjoint intervals.
By Lemma \ref{lem: wightman locality vs vertex locality}, $\varphi_{u_1}(f_1)\varphi_{u_2}(f_2)$ and $\varphi_{u_2}(f_2)\varphi_{u_1}(f_1)$ agree as maps $\cV \to \widehat \cV$.
In fact, since $\cV$ lies inside the domain $\cD$, the two agree as maps into the Hilbert space.
Thus if $g_1, \ldots, g_m$ are Laurent polynomials and $v_1, \ldots, v_m \in \cS$ we have
$$
 \varphi_{u_1}(f_1)\varphi_{u_2}(f_2)\big[ \varphi_{v_1}(g_1) \cdots \varphi_{v_m}(g_m)\Omega \big]
=
\varphi_{u_2}(f_2)\varphi_{u_1}(f_1)\big[\varphi_{v_1}(g_1) \cdots \varphi_{v_m}(g_m)\Omega\big].
$$
Since such expressions are continuous in the $g_j$, by the vacuum axiom we see that $\varphi_{u_1}(f_1)\varphi_{u_2}(f_2) = \varphi_{u_2}(f_2)\varphi_{u_1}(f_1)$ as operators on $\cD$, completing the proof.
\end{proof}

For the other direction, it is convenient to first introduce a weaker notion of Wightman fields, which will be sufficient to construct unitary vertex algebras.

\begin{defn}\label{def: quasiwightman CFT}
Let $\cH$ be a Hilbert space carrying a SOT-continuous unitary representation $U$ of $\Mob$.
Let $\cV = \bigoplus \cV(n)$ be the finite-energy vectors, where $\cV(n) = \mathrm{ker}(L_0-n)$, and assume that each $\cV(n)$ is finite-dimensional.
A \textbf{quasi-Wightman field theory} on $\cV$ is given by
a family $\cF$ of adjointable operator-valued distributions on $S^1$ with domain $\cV$, along with the representation $U$ and a vector $\Omega \in \cV$, such that the following hold: 
\begin{enumerate}[{(QW}1{)}]
\item \textbf{M\"obius covariance}: Each $\varphi \in \mathcal{F}$ is M\"obius-covariant in the Wightman sense.
\item \textbf{Adjoint}: $\cF$ is closed under $\varphi \mapsto \varphi^\dagger$.
\item \textbf{Locality}: If $f,g$ have disjoint supports, then then $\varphi_1(f)\varphi_2(g) = \varphi_2(g)\varphi_1(f)$ as maps $\cV \to \widehat \cV$.
\item \textbf{Spectrum condition}: The generator $L_0$ of rotations $U(R_\theta)=e^{i\theta L_0}$ is positive.
\item \textbf{Vacuum}: The vector $\Omega$ is the unique (up to scalar) vector in $\cV$ that is invariant under $U$,
and $\cV$ is spanned by expressions of the form%
\footnote{%
Note that $\varphi_j(z^{n_j}) \in \End(\cV)$ by M\"obius covariance.%
} 
$\varphi_1(z^{n_1}) \cdots \varphi_k(z^{n_k})\Omega$ where $\varphi_j \in \cF$ and $n_j \in \bbZ$.
\end{enumerate}
\end{defn}
Note that a family of fields satisfying all of the conditions of Definition \ref{def: quasiwightman CFT} except for the cyclicity of the vacuum yields a quasi-Wightman field theory on an appropriate subspace of $\cH$.

The difference between quasi-Wightman field theories and Wightman field theories is that in quasi-Wightman field theories we do not assume that the operators $\varphi(f)$ preserve the domain in general.
However one of our main results (Theorem \ref{thm: wightman CFT iff va}) implies that this property is automatic under the additional condition
of uniformly bounded order.

We say that a quasi-Wightman field theory $\cF$ has \textbf{uniformly bounded order}
if every $\varphi \in \cF$ has uniformly bounded order (Definition \ref{defn: uniformly bounded order}).
The following is almost immediate.
\begin{prop}\label{prop: Wightman to quasi Wightman}
Let $\cH$ be a Hilbert space carrying a positive-energy representation of $\Mob$, let $\cV$ be the finite-energy vectors, and suppose that each weight space $\cV(n)$ is finite-dimensional.
Let $\cF$ be a Wightman field theory on $\cH$.
Then $\tilde \cF = \{ \varphi|_{\cV} \, : \, \varphi \in \cF\}$ is a quasi-Wightman field theory on $\cV$.
If $\cF$ has uniformly bounded order then so does $\tilde \cF$.
\end{prop}
\begin{proof}
Observe that $\cV \subset \cD$ since $\dim \cV(n) < \infty$, as described in Section \ref{sec:wightman}.
An adjointable field remains adjointable when restricted to a dense subspace of the domain, so in particular each $\varphi|_\cV$ is adjointable.
Each condition of the definition of quasi-Wightman field theory follows directly from the corresponding property of a Wightman field theory.
\end{proof}

We now give a converse to Theorem \ref{thm: VA to Wightman} by constructing a unitary M\"obius vertex algebra from a (quasi-)Wightman field theory.

\begin{thm}\label{thm: quasi Wightman to VA}
Let $\cH$ be a Hilbert space carrying a positive-energy representation of $\Mob$, with finite-energy vectors $\cV$.
Assume that the weight spaces $\cV(n)=\mathrm{ker}(L_0-n)$ are finite-dimensional and that $\dim V(0) = 1$.
Let $\cF$ be a quasi-Wightman field theory on $\cV$ with uniformly bounded order.
Then there is a unique structure of a unitary M\"obius vertex algebra on $\cV$ such that for every $\varphi \in \cF$ there is a quasi-primary $v \in \cV$ such that $\varphi(f) = Y(v,f)|_\cV$.
\end{thm}
\begin{proof}
We note that the uniqueness of such a vertex algebra structure is guaranteed by the fact that the condition $\varphi(f) = Y(v,f)|_{\cV}$ determines the modes of the generating fields.

Recall from \eqref{eqn: formal distribution} that for a field $\varphi \in \cF$
which is M\"obius-covariant of degree $d$, the associated point-like field is $\hat \varphi(z) = \sum_{n \in \bbZ} \varphi_n z^{-n-d}$, where $\varphi_n = \varphi(z^n)$.
Let 
\[
\hat \cF = \{ \hat \varphi(z) : \varphi \in \cF\}
\]
 be the family of point-like fields associated to $\cF$.

First, we construct an $\bbN$-graded M\"obius vertex algebra structure on $\cV$ by showing that $\hat \cF$ satisfies the six hypotheses of Theorem \ref{thm: existence for Mobius}.
The first two hypotheses of Theorem \ref{thm: existence for Mobius} hold by assumption.
By Lemma \ref{lem: Mobius covariance conditions}, each field of $\cF$ is M\"obius-covariant in the vertex algebra sense, and thus the fourth hypothesis holds for $\hat \cF$.
By assumption, any pair of fields in $\cF$ is relatively local in the sense of Wightman fields.
It follows that the point-like fields in $\hat \cF$ are relatively local in the sense of vertex operators by Lemma \ref{lem: wightman locality vs vertex locality} and the uniformly bounded orders of the fields.
This shows that the fifth condition of Theorem \ref{thm: existence for Mobius} is satisfied.
The sixth condition, that $\cV$ is generated by modes of fields from $\hat \cF$, is immediate from the definition of quasi-Wightman field theory.

It thus remains to show the third condition, namely that $\hat \varphi(z)\Omega$ has a removable singularity at $z=0$.
That is, we must show that $\varphi_{-n} \Omega = 0$ for $n\le d-1$.
When $n < 0$ the positivity of $L_0$ implies that $\varphi_{-n}\Omega = 0$. 
It remains to consider the cases $n=0, 1, \ldots, d-1$.
For such $n$, the M\"obius covariance of $\hat \varphi$ \eqref{eq:quasi-primary} and the M\"obius invariance of $\Omega$ yield that
\[
\varphi_{-n}\Omega = \tfrac{1}{n-d}[L_{-1}, \varphi_{-(n-1)}]\Omega = \tfrac{1}{n-d} L_{-1}\varphi_{-(n-1)}\Omega.
\]
Iteratively applying this identity, starting with $n=0$, yields
\[
0 = \varphi_0\Omega = \varphi_{-1}\Omega = \cdots = \varphi_{-(d-1)}\Omega,
\]
as required.

By Theorem \ref{thm: existence for Mobius}, we conclude that there exists a $\bbN$-graded M\"obius vertex algebra structure on $\cV$ such that for all $\varphi \in \cF$ we have $\hat \varphi(z) = Y(v_\varphi,z)$, where $v_\varphi = \lim_{z \to 0} \hat \varphi(z)\Omega$.
It follows that $\varphi(f)u = Y(v_\varphi,f)u$ for $f$ a Laurent polynomial, and indeed for all $f \in C^\infty(S^1)$ if both sides are regarded as elements of $\widehat \cV$.
The vector $v_\varphi$ is quasi-primary as
$$
L_1 v_\varphi = \lim_{z \to 0} L_1 Y(v_\varphi,z)\Omega = \lim_{z \to 0} [L_1, Y(v_\varphi,z)]\Omega=
\lim_{z \to 0} (z^2 \tfrac{d}{dz} + 2zd_\varphi)Y(v_\varphi,z)\Omega = 0.
$$

It remains to show that $\cV$ is unitary.
The assumption that $\cF$ is invariant under $\dagger$ shows that $\tfrac{1}{2}(\hat \varphi + \widehat{\varphi^\dagger})$ and $\tfrac{1}{2i}(\hat \varphi - \widehat{\varphi^\dagger})$ are fields of the vertex algebra which are Hermitian.
Such fields generate $\cV$, and thus $\cV$ is unitary by Proposition \ref{prop: unitary iff Hermitian generating}.
\end{proof}

In summary, we have given constructions which relate the notions of unitary M\"obius vertex algebra and M\"obius-covariant (quasi-)Wightman field theories with uniformly bounded order.
All of these constructions start with a Hilbert space $\cH$ carrying a positive-energy representation of $\Mob$.
In all of the constructions, we require that the weight spaces $\cV(n)=\mathrm{ker}(L_0 - n)$ are finite-dimensional and that $\dim \cV(0) = 1$.
We write $\cV = \bigoplus_{n=0}^\infty \cV(n)$ for the space of finite-energy vectors.
The constructions are then given by the following steps:
\begin{enumerate}
\item Starting with a unitary M\"obius vertex algebra $\cV$ and a $(-1)^{L_0}\Theta$-invariant generating set of quasi-primary vectors $\cS$, Theorem \ref{thm: VA to Wightman} says that $\cF = \{\varphi_v : v \in \cS\}$ is a M\"obius-covariant Wightman field theory.
Here $\varphi_v(f)$ is the restriction of $Y(v,f)$ to the invariant domain $\cD$ generated from the vacuum by all of the smeared fields.
The resulting Wightman field theory has uniformly bounded order.
\item Given a M\"obius-covariant Wightman field theory on a domain $\cD$, one may restrict all of the distributions to the smaller domain $\cV$ of finite-energy vectors, and the result is a quasi-Wightman field theory by Proposition \ref{prop: Wightman to quasi Wightman}.
\item From a M\"obius-covariant quasi-Wightman field theory $\cF$ with uniformly bounded order, Theorem \ref{thm: quasi Wightman to VA} asserts the existence of a unique unitary M\"obius vertex algebra structure on $\cV$ such that for every $\varphi \in \cF$ there is a $v_\varphi \in \cV$ such that $\varphi(f) = Y(v_\varphi, f)|_{\cV}$.
There is a canonical $(-1)^{L_0}\Theta$-invariant generating set $\cS = \{v_\varphi : \varphi \in \cF\}$.
\end{enumerate}

One may begin with a vertex algebra or (quasi-)Wightman field theory and cycle through the above steps, returning to the original object.
We have therefore proven the following.

\begin{thm}\label{thm: wightman CFT iff va}
Let $\cH$ be a Hilbert space carrying a positive-energy representation of $\Mob$, and let $\cV = \bigoplus_{n=0}^\infty \cV(n)$ be the finite-energy vectors.
Assume that $\cV(n)$ is finite-dimensional and $\dim \cV(0) = 1$.
Then the constructions of Theorem \ref{thm: VA to Wightman}, Proposition \ref{prop: Wightman to quasi Wightman}, and Theorem \ref{thm: quasi Wightman to VA} provide bijections between the following notions:
\begin{enumerate}
\item\label{it:voa} Unitary M\"obius vertex algebra structures on $\cV$, equipped with a $(-1)^{L_0}\Theta$-invariant set of quasi-primary vectors that generate $\cV$.
\item\label{it:wightman} Wightman field theories on $\cH$ with uniformly bounded order.
\item\label{it:qwightman} Quasi-Wightman field theories on $\cV$ with uniformly bounded order.
\end{enumerate}
\end{thm}

We do not have an example of a (quasi)-Wightman field theory that does not have uniformly bounded order, and it is possible that this property is automatic.
The smeared fields arising from vertex algebras have uniformly bounded order, so one approach to demonstrating that this property is automatic for Wightman field theories would be to give a proof of Theorem \ref{thm: quasi Wightman to VA} which does not use the uniformly bounded order hypothesis.
In particular, the only place where this hypothesis is invoked is in the proof of Lemma \ref{lem: wightman locality vs vertex locality} regarding locality, so it would suffice to show that the fields of a Wightman field theory were automatically local in the vertex algebra sense.
There is a sketch of an argument demonstrating this in \cite[\S1.2]{Kac98}, but it does not seem complete%
\footnote{On page
11 the operator product expansion of two fields $[\Phi_a(t), \Phi_b(t')] = \sum_{j\ge 0} \delta^{(j)}(t-t')\Psi^j(t')$ is stated
and it is claimed that for $\Psi^j$ the Wightman axioms still hold, without proof. We are unable to justify this.
Cf.\! \cite{Bostelmann05OPE} which proves a form of OPE under additional assumptions.}%
.

\section{The VOA associated to a conformal net}\label{sec: conformal nets}
In this section, we introduce an operator-algebraic formulation of M\"obius-covariant field theory.
A \textbf{M\"obius-covariant (Haag-Kastler) net} on $S^1$ is a triple $(\cA, U, \Omega)$, where
$\cA$ associates to each open nondense nonempty connected interval $I$ of $S^1$ a von Neumann algebra $\cA(I)$ on $\cH$,
$U$ is a SOT-continuous unitary representation of $\Mob$ on $\cH$, and $\Omega \in \cH$ such that the following hold:
\begin{enumerate}[{(HK}1{)}]
\item \textbf{Isotony}: If $I_1 \subset I_2$, then $\cA(I_1) \subset \cA(I_2)$.
\item \textbf{Locality}: If $I_1 \cap I_2 = \emptyset$, then $\cA(I_1)$ and $\cA(I_2)$ commute.
\item \textbf{M\"obius covariance}: For $\gamma \in \Mob$, $U(\gamma)\Omega = \Omega$ and
$U(\gamma)\cA(I) U(\gamma)^* = \cA(\gamma I)$ for each interval $I$.
\item \textbf{Spectrum condition}: The generator $L_0$ of rotations $U(R_\theta)=e^{i\theta L_0}$ is positive.
\item \textbf{Vacuum}: $\Omega$ is the unique (up to a scalar) vector in $\cH$ that is invariant under $U$,
and it is cyclic for $\bigvee_{I\Subset S^1}\cA(I)$.
\end{enumerate}

While we will primarily be concerned with M\"obius-covariant nets, we will also briefly discuss \textbf{diffeomorphism covariant} conformal nets, which have the additional property that $U$ extends to a projective unitary representation of orientation-preserving diffeomorphisms $\mathrm{Diff}(S^1)$.
This representation is required to satisfy $U(\gamma)\cA(I)U(\gamma)^* = \cA(\gamma I)$ and moreover $U(\gamma) \in \cA(I)$ when $\gamma$ is supported in $I$.

\begin{defn}\label{def:affiliated}
Let $\cA$ be a M\"obius-covariant net with vacuum Hilbert space $\cH$, and let $\cV = \bigoplus \cV(n)$
be the space of finite-energy vectors.
An adjointable distribution $\varphi$ with domain $\cV$ is \textbf{affiliated with $\cA$} if whenever $\supp(f) \subset I$ there is an intermediate extension $\varphi(f) \subseteq \varphi^{\text{aff}}(f) \subseteq \varphi^\dagger(\overline{f})^*$ such that $\varphi^{\text{aff}}(f)$ is affiliated with $\cA(I)$.
\end{defn}

For a discussion on unbounded operators affiliated with a von Neumann algebra, see \cite[\S5.2]{Pedersen89}.
We note that the requirement that $\varphi^\dagger(\overline{f})^*$ is an extension of $\varphi^{\text{aff}}(f)$ is the same as requiring that the domain of $\varphi^{\text{aff}}(f)^*$ contains $\cV$.

We will show that the M\"obius-covariant fields affiliated with $\cA$ satisfy all of the axioms of a quasi-Wightman field theory except perhaps cyclicity of the vacuum (which in turn produces a quasi-Wightman field theory on a subspace of the finite-energy vectors).
The main step is to check the locality axiom, for which we will use the following basic lemma.

\begin{lem}\label{lem: approximating affiliated operators}
Let $X$ be a densely defined closed operator on a Hilbert space $\cH$, and suppose that $X$ is affiliated with a von Neumann algebra $\cM \subset \cB(\cH)$.
Then there exists a sequence $x_n \in \cM$ such that $x_n\xi \to X\xi$ for all $\xi \in \dom(X)$ and $x_n^*\xi \to X^*\xi$ for all $\xi \in \dom(X^*)$.
\end{lem}
\begin{proof}
Let $f_n:[0,\infty) \to [0,\infty)$ be the bounded function $f_n(t) = \min(n,t)$.
Let $X = v\abs{X}$ be the polar decomposition, with $v$ supported on $\overline{\Ran(|X|)}$, and set $x_n = v f_n(\abs{X}) \in \cM$.
For any $\xi \in \dom(X) = \dom(\abs{X})$ we have $f_n(\abs{X})\xi \to \abs{X}\xi$, and thus $x_n\xi \to X\xi$.
We have $X^* = (v\abs{X})^* = \abs{X}v^*$ since $v$ is bounded, and thus $v^*$ maps $\dom(X^*)$ into $\dom(\abs{X})$.
Thus if $\xi \in \dom(X^*)$ we have 
$$
x_n^*\xi = f_n(\abs{X})v^*\xi \to \abs{X}v^*\xi = X^*\xi
$$
as desired.
\end{proof}

\begin{thm}\label{thm: affiliated fields are Wightman fields}
Let $(\cA,U,\Omega)$ be a M\"obius-covariant net and let $\cF$ be the family of M\"obius-covariant adjointable distributions with domain $\cV$ that are affiliated with $\cA$.
Then $\cF$ satisfies all of the axioms of a quasi-Wightman field theory except perhaps for the cyclicity of the vacuum.
\end{thm}
\begin{proof}
We must check that $\cF$ satisfies Definition \ref{def: quasiwightman CFT}.
All of the conditions are immediate except closure under $\dagger$ and the locality axiom.

Let $\varphi \in \cF$.
We must show that $\varphi^\dagger$ is affiliated with $\cA$ and that it is M\"obius-covariant.
Fix a function $f \in C^\infty(S^1)$ that is supported in $I$.
Since $\varphi$ is affiliated with $\cA$ we have extensions
$$
\varphi(\overline{f}) \subseteq \varphi^{\text{aff}}(\overline{f}) \subseteq \varphi^\dagger(f)^*
$$
with $\varphi^{\text{aff}}(\overline{f})$ affiliated with $\cA(I)$.
Taking adjoints (and forgetting the closure of $\varphi^\dagger(f)$) we have
$$
\varphi^\dagger(f) \subseteq \varphi^{\text{aff}}(\overline{f})^* \subseteq \varphi(\overline{f})^*.
$$
Since the adjoint $\varphi^{\text{aff}}(\overline{f})^*$ is an extension of $\varphi^\dagger(f)$ and is affiliated with $\cA(I)$,
we have shown that $\varphi^\dagger$ is again affiliated with $\cA$ (see Definition \ref{def:affiliated}).

We now check M\"obius covariance of $\varphi^\dagger$.
By Lemma \ref{lem: Mobius covariance conditions} we may check that $\varphi^\dagger$ is M\"obius-covariant in the vertex algebra sense, namely that for all $g \tfrac{d}{d\theta} \in \Lie(\Mob)_\bbC$ we have
$$
[\pi(g \tfrac{d}{d \theta}), \varphi^\dagger(f)]u = \varphi^\dagger((d-1)\tfrac{dg}{d\theta} f - g \tfrac{d f}{d\theta})u,
$$
where $\pi$ is the representation of $\Lie(\Mob)_\bbC$ that integrates to $U$.
Both sides of this condition depend continuously on $f$ (in the topology of $\widehat \cV$) and thus it suffices to consider when $f$ is a Laurent polynomial.
By M\"obius covariance for $\varphi$ (applied to $\overline{f}$ and $\overline{g}$) we have
$$
[\pi(\overline{g} \tfrac{d}{d \theta}), \varphi(\overline{f})]u = \varphi(\overline{(d-1)\tfrac{dg}{d\theta} f - g \tfrac{d f}{d\theta}})u.
$$
Taking adjoints as endomorphisms of $\cV$, and applying the fact that $\pi(\overline{g} \tfrac{d}{d \theta})^* = -\pi(g \tfrac{d}{d \theta})$, we obtain the M\"obius covariance condition for $\varphi^\dagger$, which completes the proof that $\cF$ is $\dagger$-invariant.

We now show that $\cF$ satisfies locality.
Let $\varphi,\psi \in \cF$, let $I_1$ and $I_2$ be disjoint intervals, and suppose that $f$ is supported in $I_1$ and $g$ is supported in $I_2$.
Choose intermediate extensions
$$
\varphi(f) \subseteq \varphi^{\text{aff}}(f) \subseteq \varphi^\dagger(\overline{f})^*, \qquad
\psi(g) \subseteq \psi^{\text{aff}}(g) \subseteq \psi^\dagger(\overline{g})^*
$$
with $\varphi^{\text{aff}}(f)$ and $\psi^{\text{aff}}(g)$ affiliated with $\cA(I)$ and $\cA(J)$, respectively.
Applying Lemma \ref{lem: approximating affiliated operators}, we may choose a sequence $x_n \in \cA(I)$ such that for every $u \in \cV$ we have $x_n u \to \varphi(f)u$ and and $x_n^*u \to \varphi(f)^*u = \varphi^\dagger(\overline{f})u$.
We also choose $y_n \in \cA(J)$ that approximate $\psi(g)$ in the same manner, and observe that the $x_n$ and $y_n$ commute by the locality axiom for $\cA$.
Then for $u,v \in \cV$ we have
\begin{align*}
\ip{&\varphi(f)\psi(g)u,v}
=
\ip{\psi(g)u,\varphi^\dagger(\overline{f})v}
=
\lim_{n \to \infty} \ip{y_n u, x_n^* v} =\\
&=
\lim_{n \to \infty} \ip{x_n u, y_n^* v}
=
\ip{\varphi(f)u,\psi^\dagger(\overline{g})v}
=
\ip{\psi(g)\varphi(f)u,v}
\end{align*}
which establishes locality, and completes the proof.
\end{proof}

Note that any $\dagger$-invariant subset of a quasi-Wightman field theory (perhaps missing cyclicity of the vacuum) is again a quasi-Wightman field theory (perhaps missing cyclicity of the vacuum).
In order to apply Theorem \ref{thm: quasi Wightman to VA}, we are particularly interested in the sub-theory of uniformly bounded distributions.

\begin{defn}\label{def: quasi Wightman fields defined from net}
Let $\cA$ be a M\"obius-covariant net on a Hilbert space $\cH$, and let $\cV \subset \cH$ be the space of finite-energy vectors.
We denote by  $\cF_\cA$ the collection of all adjointable M\"obius-covariant distributions $\varphi$ with domain $\cV$ that are affiliated with $\cA$ and have the property that both $\varphi$ and $\varphi^\dagger$ have uniformly bounded order.
\end{defn}

By Theorem \ref{thm: affiliated fields are Wightman fields}, $\cF_\cA$ satisfies all of the axioms of a quasi-Wightman field theory except perhaps cyclicity of the vacuum, and thus yields a quasi-Wightman field theory on 
$$
\cV_\cA = \spann \set{\varphi_1(z^{n_1}) \cdots \varphi_k(z^{n_k})\Omega}{ \varphi_i \in \cF_\cA, n_i \in \bbZ} \subseteq \cV,
$$
provided that the $L_0$-eigenspaces are finite-dimensional.
This quasi-Wightman field theory has uniformly bounded order by construction, and so by Theorem \ref{thm: wightman CFT iff va} there is a unitary M\"obius vertex algebra structure on $\cV_\cA$.
In general, we cannot guarantee that $\cV_\cA$ is not the trivial vertex algebra $\{\Omega\}$.
However, if $\cA$ is a conformal net (i.e.\! if $\cA$ is diffeomorphism covariant), then the associated stress-energy tensor $T(z) = \sum_{n \in \bbZ} L_n z^{-n-2}$ induces a quasi-Wightman field affiliated with $\cA$ \cite[Appendix]{Carpi04}.

\begin{cor}\label{cor: conformal net to uvoa}
Let $\cA$ be a diffeomorphism covariant conformal net, and let $T(z) = \sum_{n \in \bbZ} L_n z^{-n-2}$ be the associated stress-energy tensor.
Suppose that the eigenspaces of $L_0$ are finite-dimensional.
Then $\cV_\cA$ is a unitary vertex operator algebra with conformal vector $\nu = L_{-2}\Omega$ and $Y(\nu,z) = T(z)$.
\end{cor}

The following proposition shows that if the net $\cA$ ``comes from'' a M\"obius vertex algebra (in a relatively weak sense), then that vertex algebra is the one we have constructed.

\begin{prop}\label{prop: universality of VA from net}
Let $\cA$ be a M\"obius-covariant net defined on a Hilbert space $\cH$ with finite-energy vectors $\cV$.
Suppose that there is a structure of a unitary M\"obius vertex algebra $Y^{\cV}$ on $\cV$, with the same vacuum vector and representation of $\Mob$ as $\cA$.
Suppose that there is a generating family of quasi-primary vectors $u$ for which the operator-valued distributions $f \mapsto Y^{\cV}(u,f)|_\cV$ are affiliated with $\cA$ in the sense of Definition \ref{def:affiliated}.
Then $\cV = \cV_\cA$ as unitary M\"obius vertex algebras.
\end{prop}
\begin{proof}
As in Theorem \ref{thm: affiliated fields are Wightman fields}, we write $\cF$ for the collection of M\"obius-covariant, adjointable distributions with domain $\cV$ that are affiliated with $\cA$, which form a quasi-Wightman field theory except perhaps for cyclicity of the vacuum.
The vertex algebra $\cV_\cA$ corresponds to the subcollection $\cF_{\cA} \subset \cF$ consisting of distributions $\varphi$ such that $\varphi$ and $\varphi^\dagger$ have uniformly bounded order.

Let $\cS \subset \cV$ be a generating set of quasi-primary vectors for the vertex algebra $\cV$ such that the associated operator-valued distributions are affiliated with $\cA$.
For $u \in \cS$, let $\varphi_u(f) = Y^{\cV}(u,f)|_\cV$ be the associated operator-valued distribution.
By Theorem \ref{thm: VA to Wightman}, the distributions $\varphi_u$ are adjointable and M\"obius-covariant, and by assumption the $\varphi_u$ are affiliated with $\cA$.
Hence $\varphi_u \in \cF$.
Moreover, by Proposition \ref{prop: VOA has unif bdd order} the fields $\varphi_u$ and $\varphi_u^\dagger$ have uniformly bounded order, so $\varphi_u \in \cF_{\cA}$.
It follows that $\cS \subset \cV_{\cA}$ and $Y^{\cV}(u,z)|_{\cV_\cA} = Y^{\cA}(u,z)$ for $u \in \cS$, where $Y^\cA$ denotes the vertex algebra structure on $\cV_\cA$ arising from our construction.
Since $\cS$ generates $\cV$, we have $\cV_{\cA} = \cV$ as vertex algebras.
The representations of $\Lie(\Mob)$ and the inner products of $\cV_\cA$ and $\cV$ coincide by construction.
\end{proof}

Finally, as a result of our analysis of the domains of smeared vertex operators, we are able to give a construction of conformal nets from vertex algebras under looser assumptions than \cite{CKLW18}.
The key hypothesis will be the following.

\begin{defn}
A unitary M\"obius vertex algebra $\cV$ is called \textbf{AQFT-local} if for any $u,v \in \cV$ the closed operators $Y(u,f)$ and $Y(v,g)$ commute strongly\footnote{%
Recall that two closed operators $X$ and $Y$ commute strongly when the von Neumann algebras $\mathrm{vN}(X)$ and $\mathrm{vN}(Y)$ commute.
Here, $\mathrm{vN}(X)$ is the smallest von Neumann algebra to which $X$ is affiliated; it is generated by the polar partial isometry of $X$ along with the spectral projections of $\abs{X}$.
We avoid the terminology ``strongly local'' because in \cite{CKLW18} ``strongly local' vertex algebras are assumed to satisfy polynomial energy bounds.
}
whenever $f$ and $g$ have disjoint support.
\end{defn}
Here $Y(u,f)$ is the closure of $Y^0(u,f)$, where the domain of $Y^0(u,f)$ is $\cV$.
We have the following analog of \cite[Theorem 6.8]{CKLW18}.
\begin{prop}
Let $\cV$ be a unitary M\"obius vertex algebra that is AQFT-local.
Then 
$$
\cA_{\cV}(I) = \mathrm{vN}( \{ Y(v,f) \, : \, \mathrm{supp}(f) \subseteq I \} )
$$
defines a M\"obius-covariant net, with M\"obius symmetry given by the representation of $\Mob$ obtained by integrating the representation of $\Lie(\Mob)$ on $\cV$.
If $\cV$ is a unitary vertex operator algebra%
\footnote{For definitions and details of unitary vertex operator algebras, as opposed to unitary M\"obius vertex algebras, see \cite[\S4,\S5]{CKLW18}}
with conformal vector $\nu$, then $\cA_\cV$ is diffeomorphism covariant with respect to the representation of $\mathrm{Diff}(S^1)$ obtained by integrating the coefficients of $Y(\nu,z)$.
\end{prop}
\begin{proof}
In the M\"obius case, the only point to address is M\"obius covariance.
Let $U:\Mob \to \cU(\cH)$ be the representation obtained by integrating the $L_k$ of $\cV$.
By the relation $Y(L_{-1}v,f) = Y(v,if^\prime - d_vf)$, we see that $\cA_\cV(I)$ is generated by smeared fields $Y(v,f)$ with $v$ quasi-primary.
Let $\cS \subseteq \cV$ be the set of quasi-primary vectors.
Then the domain $\cD = \{ Y(v^k,f_k) \cdots Y(v^1,f_1)\Omega \, : \, v^j \in \cS, \, f_j \in C^\infty(S^1) \}$ contains $\cV$, and is contained in the domain of any $Y(v,f)$ by Proposition \ref{prop: VOA has unif bdd order}.
Hence $\cD$ is a core for $Y(v,f)$.
By Theorem \ref{thm: VA to Wightman}, we have $U(\gamma)\cD = \cD$ for any $\gamma \in \Mob$, and $U(\gamma)Y(v,f)|_{\cD} U(\gamma)^* = Y(v,\beta_d(\gamma)f)|_{\cD}$ when $v$ is quasi-primary with conformal dimension $d$.
Taking closures yields $U(\gamma)Y(v,f)U(\gamma)^* = Y(v,\beta_d(\gamma)f)$, and M\"obius covariance follows.
If $\cV$ is a VOA, then (as in \cite[Theorem 6.8]{CKLW18}) we observe that $\cA_\cV$ is an extension of a Virasoro net and its conformal covariance follows from \cite[Proposition 3.7]{Carpi04}.
\end{proof}

\begin{cor}
Let $\cV$ be an AQFT-local unitary M\"obius vertex algebra, and let $\cA_\cV$ be the corresponding M\"obius-covariant net.
Then the unitary M\"obius vertex algebra associated to $\cA_\cV$ via Corollary~\ref{cor: conformal net to uvoa} recovers $\cV$.
\end{cor}
\begin{proof}
For every quasi-primary $u \in \cV$ the operator $Y(u,f)$ is affiliated with $\cA_{\cV}(I)$ when $\supp f \subseteq I$.
Thus the claimed result is an immediate consequence of Proposition~\ref{prop: universality of VA from net}.
\end{proof}

There is another construction of conformal nets $\cA_{\cV}$ from unitary vertex operator algebras $\cV$ which have a property called ``bounded localized vertex operators''  \cite{Tener19Geometric, Tener19Representation, Tener19Fusion}.
A straightforward argument shows that if $\cV$ has this property then applying Corollary~\ref{cor: conformal net to uvoa} to $\cA_{\cV}$ recovers $\cV$, but we will omit a discussion of bounded localized vertex operators here.

Work in progress of Andr\'e Henriques and JT will show that every (diffeomorphism covariant) conformal net is of the form $\cA_\cV$, for both the AQFT-locality construction and the bounded localized vertex operator construction.
It would then follow that the unitary VOA $\cV_\cA$ from Corollary \ref{cor: conformal net to uvoa} is always defined on all of $\cV$, and moreover that $\cV_\cA$ is AQFT-local with corresponding conformal net $\cA$.

It would be interesting to know whether the dual nets of the Virasoro net with $c>1$ \cite{BS90} or
the $\mathcal{W}_3$-net with $c>2$ \cite{CTW21} come from vertex algebras. It is unknown whether these dual nets
are diffeomorphism covariant.

\appendix

\section{Generating unitary M\"obius vertex algebras}

The following theorem gives standard criteria for a collection of fields to generate a M\"obius vertex algebra.
However, standard references (e.g.\! \cite[Theorem 4.5]{Kac98}) provide a proof under additional hypotheses and do not treat $L_1$-covariance. We summarize the necessary adjustments here:

\begin{thm}\label{thm: existence for Mobius}
Let $\cV$ be a vector space equipped with a representation $\{L_{-1},L_0,L_1\}$ of $\Lie(\Mob)_\bbC$, a collection $\cF \subset \End(\cV)[[z^{\pm 1}]]$ of formal distributions, and a choice of vector $\Omega \in \cV$ such that the following hold:
\begin{enumerate}
\item $\cV = \bigoplus_{n=0}^\infty \cV(n)$ where $\cV(n) =  \ker(L_0 - n)$ and each $\cV(n)$ is finite-dimensional, with $\dim \cV(0) = 1$.
\item $\Omega$ is $\Lie(\Mob)$-invariant.
\item For every $A(z) \in \cF$, $A(z)\Omega$ has a removable singularity at $z=0$.
\item For every $A(z) \in \cF$, there exists a number $d_A \in \bbZ_{\ge 0}$ such that
$$
[L_{k}, A(z)] = (z^{k+1} \tfrac{d}{dz} + (k+1)z^k d_A) A(z) \qquad k=-1,0,1.
$$
\item For every $A(z),B(w) \in \cF$, $(z-w)^N[A(z),B(w)] = 0$ for $N$ sufficiently large.
\item $\cV = \operatorname{span} \{ A^1_{(n_1)} \cdots A^k_{(n_k)}\Omega \,\,: \,\, A^j \in \cF,\, n_j \in \bbZ,\, k \in \bbZ_{\ge0}\}$, where $A_{(n)}$ denotes the $n$th mode of the formal series $A(z)$.
\end{enumerate}

Then there exists a unique structure of an $\bbN$-graded M\"obius vertex algebra on $\cV$ such that $A(z) = Y(A(z)\Omega|_{z=0},z)$ for every $A \in \cF$.

\end{thm}
\begin{proof}
Note that the uniqueness part of the conclusion is clear by the generating property 6 along with the Borcherds product formula \cite[Eq. (4.6.10)]{Kac98}.
Let $\cG$ be the collection of formal distributions obtained by closing $\cF \cup \{\operatorname{Id}\}$ under derivatives, Borcherds $(n)$-products 
$$\label{eq:nproducts}
A_{(n)}B(z) = \operatorname{Res}_{z} \left\lbrace (z-w)^nA(z)B(w) - (-w+z)^nB(w)A(z)\right\rbrace
$$
for $n \in \bbZ$, and linear combinations. Here $(z-w)^n$ is shorthand for the series expansion of that function in the domain $\abs{z}>\abs{w}$, and similarly $(-w+z)^n$ is expanded in the domain $\abs{w} > \abs{z}$.
As in the proof of \cite[Theorem 4.5]{Kac98}, $\cG$ again satisfies properties 3 and 5 of the hypothesis, as well as the $L_{-1}$-commutation relations of 4.
We note that Kac only considers $(n)$-products with $n > 0$, but that does not affect this step (most notably Dong's Lemma \cite[Lemma 3.2]{Kac98} has no restriction on $n$).
By hypothesis 6, the map $\cG \to \cV$ given by $A(z) \mapsto A(z)\Omega|_{z=0}$ is surjective.
The argument of the uniqueness theorem \cite[Theorem 4.4]{Kac98}\cite{Goddard89}, shows that this map is injective. By inverting, we obtain the state-field correspondence $Y$ of a vertex algebra such that $Y(A(z)\Omega|_{z=0},z) = A(z)$ for every $A \in \cF$.

It remains to check that the fields of this vertex algebra have the correct commutation relations with $L_0$ and $L_1$.
From the $L_0$-commutation relation for $\cF$ and the definition of the $(n)$-product, it is immediate to check the $L_0$-commutation relation for $\cG$, and that $A(z)\Omega|_{z=0} \in \cV(d_A)$.
To address the $L_1$-commutation relation, observe that if $A(z) \in \cF$, then 
$$
L_1 A(z)\Omega|_{z=0} = [L_1,A(z)]\Omega\mid_{z=0} = (z^2 \tfrac{d}{dz} + 2d_A z)A(z)\Omega\mid_{z=0} = 0.
$$
Thus, the states $A(z)\Omega|_{z=0}$ are quasi-primary, and so the fields $Y(A(z)\Omega|_{z=0},z)$ satisfy the $L_1$-commutation relation when $A(z) \in \cF$.
The final step is to show that the $L_1$-commutation relation extends from the generating set to the entire vertex algebra, which is the content of Lemma \ref{lem: L1 closed under Borcherds products} below.
\end{proof}

\begin{lem}\label{lem: L1 closed under Borcherds products}
Suppose that $\cV$ satisfies all of the axioms of a M\"obius vertex algebra, except perhaps the commutation relation with $L_1$.
Then the set of homogeneous vectors $v \in \cV$ such that the $L_1$-commutation relation
\begin{align}\label{eq:L1base}
[L_1, Y(v,z)] = (z^2\partial_{z} + 2d_{v}z) Y(v,z) + Y(L_{1}v,z)
\end{align}
holds is closed under Borcherds $(n)$-products $u_{(n)}v$. That is, if $Y(u,z)$ and $Y(v,z)$ satisfy the $L_{1}$-commutation relation, then
\begin{equation}
\begin{aligned}
\label{eq:lemma}
    \comm{L_{1}}{Y(u_{(n)}v,z)} &= z^2\partial_{z}Y(u_{(n)}v,z) + 2d_{u_{(n)}v}zY(u_{(n)}v,z) + Y(L_{1}u_{(n)}v,z),
\end{aligned}
\end{equation}
where $d_{u_{(n)}v} = d_u + d_v - n - 1$.
\end{lem}
\begin{proof}
We make use of the formula (see \cite[Eq (3.1.12) and (4.6.10)]{Kac98})
\begin{equation}\label{LLformula}
    \begin{aligned}
    Y(u_{(n)}v,w) = \operatorname{Res}_{z}\left\lbrace(z-w)^nY(u,z)Y(v,w) - (-w+z)^nY(v,w)Y(u,z) \right\rbrace,
    \end{aligned}
\end{equation}
where as before $(z-w)^n$ is shorthand for the series expansion of that function in the domain $\abs{z}>\abs{w}$, and similarly $(-w+z)^n$ is expanded in the domain $\abs{w} > \abs{z}$.
Taking the commutator with $L_{1}$, we have
\begin{equation}\label{eq:commexpanded}
    \begin{aligned}
   \comm{L_{1}}{Y(u_{(n)}v,w)} = \operatorname{Res}_{z}&\left\lbrace(z-w)^n(z^2\pd_{z} + w^2 \pd_{w} + 2d_{u}z + 2d_{v}w )Y(u,z)Y(v,w) \right. \\
   &-\left. (-w+z)^n(z^2\pd_{z} + w^2 \pd_{w} + 2d_{u}z + 2d_{v}w )Y(v,w)Y(u,z) \right. \\
   &+\left. (z-w)^n(Y(L_{1}u,z)Y(v,w) + Y(u,z)Y(L_{1}v,w)) \right.\\
   &- \left. (-w+z)^n(Y(L_{1}v,w)Y(u,z)+ Y(v,w)Y(L_{1}u,z)) \right\rbrace.
    \end{aligned}
\end{equation}
The terms in the final two lines yield $Y((L_{1}u)_{(n)}v,w)$ and $Y(u_{(n)}(L_{1}v),w)$. 
Applying commutation relations, we have that $L_{1}u_{(n)} = (2d_{u} -n -2)u_{(n+1)} + (L_{1}u)_{(n)} + u_{(n)}L_{1}$, and hence
\begin{equation*}
    \begin{aligned}
    Y(L_{1}u_{(n)}v,w) = (2d_{u} -n -2)Y(u_{(n+1)}v,w) + Y((L_{1}u)_{(n)}v,w)+Y(u_{(n)}(L_{1}v),w).
    \end{aligned}
\end{equation*}
Applying this equality, we can expand terms which appear on the righthand side of \eqref{eq:lemma} using \eqref{LLformula} as follows:
\begin{equation*}
\begin{aligned}
    w^{2}\partial_{w}Y(u_{(n)}&v,w) \\
    = \operatorname{Res}_{z} &\left\lbrace -w^2 n (z-w)^{n-1}Y(u,z)Y(v,w)  - (-1)w^2 n(-w+z)^{n-1}Y(v,w)Y(u,z)  \right.\\
    &  \left. +(z-w)^n w^2 Y(u,z)\partial_{w}Y(v,w) - (-w+z)^n w^2 \partial_{w}Y(v,w)Y(u,z) \right\rbrace,
\end{aligned}
\end{equation*}
and similarly,
\begin{equation*}
\begin{aligned}
    2(d_{u} +d_{v}& -n-1)wY(u_{(n)}v,w) + (2d_{u} - n - 2)Y(u_{(n+1)}v,w) \\ 
    = \operatorname{Res}_{z}&\left\lbrace (2d_{u}z + 2d_{v}w -nz-nw-2z)(z-w)^{n}Y(u,z)Y(v,w) \right. \\
    - & \,\left. (2d_{u}z + 2d_{v}w -nz-nw-2z)(-w+z)^{n}Y(v,w)Y(u,z) \right\rbrace.
\end{aligned}
\end{equation*}
Using the identity
\begin{equation*}
    \begin{aligned}
        w^2n(z-w)^{n-1} +n(z+w)(z-w)^{n} = z^2n(z-w)^{n-1},
    \end{aligned}
\end{equation*}
and applying integration-by-parts to the expression
\begin{equation*}
    \begin{aligned}
    -\operatorname{Res}_{z}&\left\lbrace (2z(z-w)^n + z^2 n (z-w)^{n-1})Y(u,z)Y(v,w) \right. \\
        & \quad - \left. (2z(-w+z)^n + z^2 n (-w+z)^{n-1})(-w+z)^n Y(v,w)Y(u,z) \right\rbrace \\
        = \operatorname{Res}_{z}&\left\lbrace - \partial_{z}(z^2(z-w)^{n})Y(u,z)Y(v,w)+\partial_{z}(z^2(-w+z)^n) Y(v,w)Y(u,z) \right\rbrace \\ 
    = \operatorname{Res}_{z}&\left\lbrace(z-w)^n z^2 \partial_{z}Y(u,z)Y(v,w) - (-w+z)^n z^2 Y(v,w)\partial_{z}Y(u,z) \right\rbrace, \\
    \end{aligned}
\end{equation*}
we have that \eqref{eq:commexpanded} is equal to the right-hand side of \eqref{eq:lemma}.
\end{proof}

\section{Remarks on Sobolev spaces}\label{sec:sobolevdist}
Recall that $H^N(S^1)$ can be defined as
\[
 H^N(S^1) = \Big\{ f \in L^2(S^1): \sum_{n \in \bbZ} (1+n^2)^N |\hat f(n)|^2 < \infty \Big\}.
\]
The norm $\|f\|_{N} = (1+n^2)^N |\hat f(n)|^2$ is equivalent to the following norm
\[
 \|f\|_{N_+} =  \sum_{n \in \bbZ} (1+n^2 + \cdots + n^{2N}) |\hat f(n)|^2,
\]
and accordingly the Sobolev space can be written as follows:
\[
 H^N(S^1) = \Big\{ f \in L^2(S^1): \sum_{n \in \bbZ} (1+n^2 + \cdots + n^{2N}) |\hat f(n)|^2 < \infty \Big\}.
\]
Using this norm, the map $U$ such that $\widehat{(Uf)}(n) = \sqrt{1+n^2 + \cdots + n^{2N}} \hat f(n)$
is a unitary from $H^N(S^1)$ to $L^2(S^1)$ (with $H^N(S^1)$ equipped with the corresponding inner product).
Any bounded operator $x$ on $L^2(S^1)$ can be written as 
$\widehat{xf}(m) = \sum_n x_{m,n} \hat f(n)$.
As it is bounded, there is a constant $C>0$ such that $|x_{m,n}|< C$.
Correspondingly, any bounded operator $x$ on $H^N(S^1)$ can be represented by components
$x_{m,n}$ and they satisfy $|x_{m,n}(1+n^2 + \cdots + n^{2N})^{-\frac12}(1+m^2 + \cdots + m^{2N})^{-\frac12}|< C$
for some $C > 0$.

Similarly, $H^{2N+2}(S^1\times S^1)$ is a subspace of $L^2(S^1\times S^1)$ with the norm
\[
 \|h\|_{{2N+2}_+(S^1\times S^1)}^2 = \sum_{m,n} \sum_{k \le 2N+2} \sum_{\ell = 0}^k m^{2(k-\ell)}n^{2\ell} |\hat h(m,n)|^2
\]
and the map $\tilde U$ such that
$\widehat{(\tilde Uf)}(m,n) = \sqrt{\sum_{k \le 2N+2} \sum_{\ell = 0}^k m^{2(k-\ell)}n^{2\ell} } \hat f(m,n)$
is a unitary from $H^{2N+2}(S^1\times S^1)$ to $L^2(S^1\times S^1)$.
Any continuous linear functional $\xi$ on $H^{2N+2}(S^1\times S^1)$
has the form $\xi(f) = \sum \xi_{m,n} \hat f(m,n)$,
where $\xi_{m,n}\left(\sum_{k \le 2N+1} \sum_{\ell = 0}^k m^{2(k-\ell)}n^{2\ell}\right)^{-\frac12}$ is in $\ell^2(\bbZ\times \bbZ)$.

Any continuous linear functional $\xi$ on $H^{2N+2}(S^1\times S^1)$
can be written as a linear combination of derivatives of $f \in L^2(S^1\times S^1)$ of order no greater than $2N+2$
\cite[Theorem 11.62]{Leoni17}.
Hence, $\xi$ itself has order no greater than $2N+4$ as a distribution on $S^1\times S^1$.

\begin{lem}\label{lem:sobolevfunctional}
 Let $B(\cdot,\cdot)$ be a jointly continuous bilinear form on $H^N(S^1)$.
 Then there is a continuous linear functional $\tilde B$ on $H^{2N+2}(S^1\times S^1)$ such that
 $\tilde B(f\otimes g) = B(f,g)$, where $(f\otimes g)(w,z) = f(w)g(z)$.
\end{lem}
\begin{proof}
 As $B$ is jointly continuous, by the Riesz representation theorem there is a bounded operator $x$ such that
 $V(f,g) = \ip{Jf, xg}$, where $\widehat {Jf}(n) = \overline{f(n)}$ is a unitary conjugation
 and $x$ can be represented by matrix components as above, satisfying
 \[
  |x_{m,n}(1+n^2 + \cdots + n^{2N})^{-\frac12}(1+m^2 + \cdots + m^{2N})^{-\frac12}|< C
 \]
 for some $C > 0$.

 We claim that
 \[
  (1+n^2 + \cdots + n^{2N})^{\frac12}(1+m^2 + \cdots + m^{2N})^{\frac12}\left(\sum_{k \le 2N+2} \sum_{\ell = 0}^k m^{2(k-\ell)}n^{2\ell}\right)^{-\frac12}
 \]
 is in $\ell^2(\bbZ\times \bbZ)$. Indeed, we only have to show that 
 \[
  (1+n^2 + \cdots + n^{2N})(1+m^2 + \cdots + m^{2N})\left(\sum_{k \le 2N+2} \sum_{\ell = 0}^k m^{2(k-\ell)}n^{2\ell}\right)^{-1}
 \]
 is in $\ell^1(\bbZ\times \bbZ)$, and note that
 \[
  \sum_{k \le 2N+2} \sum_{\ell = 0}^k m^{2(k-\ell)}n^{2\ell} > (1+n^2 + \cdots + n^{2(N+1)})(1+m^2 + \cdots + m^{2(N+1)}),
 \]
 while
 \[
  \frac{(1+n^2 + \cdots + n^{2N})(1+m^2 + \cdots + m^{2N})}{(1+n^2 + \cdots + n^{2(N+1)})(1+m^2 + \cdots + m^{2(N+1)})} < \frac{1}{n^2 m^2}
 \]
 is in $\ell^1(\bbZ\times \bbZ)$.

 Therefore, we have that
 \[
  x_{m,n}\left(\sum_{k \le 2N+2} \sum_{\ell = 0}^k m^{2(k-\ell)}n^{2\ell}\right)^{-\frac12}  \in \ell^2(\bbZ\times \bbZ)
 \]
 and $x_{m,n}$ defines a continuous linear functional $\tilde B$ on $H^{2N+2}(S^1\times S^1)$ such that
 $\tilde B(f,g) = \sum_{m,n} x_{m,n}\hat f(m)\hat g(n)$, which is an extension of the desired form. 
\end{proof}

{\small
\def\cprime{$'$} \def\polhk#1{\setbox0=\hbox{#1}{\ooalign{\hidewidth
  \lower1.5ex\hbox{`}\hidewidth\crcr\unhbox0}}} \def\cprime{$'$}

}

\end{document}